\documentclass{article}

\usepackage[backend=biber]{biblatex}
\usepackage{auxhook}
\usepackage{amsmath}
\usepackage{algorithm2e}
\usepackage{tabularx}
\usepackage{todonotes}
\usepackage{amsthm}
\usepackage{adjustbox}
\usepackage{array}
\usepackage{color}
\usepackage{colortbl}
\usepackage{multirow}
\usepackage{amsfonts}
\usepackage{amssymb}

\newcommand{\OUT}[1]{}
\newcommand{\eGDN}{eGDN}
\newcommand{\powset}[1]{\mathcal{P}(#1)}
\newcommand{\citeOwn}[1]{\textbf{\cite{#1}}}
\newcommand{\requirement}[1]{\textbf{R\,#1}\xspace}

\newcommand{\rot}[1]{\multicolumn{1}{c}{\adjustbox{angle=60,lap=\width-1em}{#1}}}

\newtheorem{theorem}{Theorem}
\newtheorem{corollary}{Corollary}

\SetKwProg{myproc}{Procedure}{}{}
\SetKwInOut{Input}{Input}
\SetKwInOut{Output}{Output}
\SetKwRepeat{Do}{do}{while}

\begin{document}

\title{Modular and Incremental Global Model Management with Extended Generalized Discrimination Networks}

\author{Matthias Barkowsky \\ matthias.barkowsky@hpi.de \and Holger Giese \\ holger.giese@hpi.de}

\maketitle

\begin{abstract}
Complex projects developed under the paradigm of model-driven engineering nowadays often involve several interrelated models, which are automatically processed via a multitude of model operations. Modular and incremental construction and execution of such networks of models and model operations are required to accommodate efficient development with potentially large-scale models. The underlying problem is also called Global Model Management.

In this report, we propose an approach to modular and incremental Global Model Management via an extension to the existing technique of Generalized Discrimination Networks (GDNs). In addition to further generalizing the notion of query operations employed in GDNs, we adapt the previously query-only mechanism to operations with side effects to integrate model transformation and model synchronization. We provide incremental algorithms for the execution of the resulting extended Generalized Discrimination Networks (\eGDN{}s), as well as a prototypical implementation for a number of example \eGDN{} operations.

Based on this prototypical implementation, we experiment with an application scenario from the software development domain to empirically evaluate our approach with respect to scalability and conceptually demonstrate its applicability in a typical scenario. Initial results confirm that the presented approach can indeed be employed to realize efficient Global Model Management in the considered scenario.
\end{abstract}


\section{Introduction}

Complex projects developed under the model-driven engineering paradigm nowadays often involve several interrelated models, which are inspected, analyzed, transformed, and synchronized via a multitude of model operations \citeOwn{seibel2010dynamic}\footnote{Note that references in bold refer to our own publications.}. An effective and efficient management of the resulting sophisticated networks of model operations is both a crucial prerequisite to successful development projects and a challenging research problem, known as Global Model Management \cite{bezivin2004modeling}.

On the one hand, modular and incremental construction of model operation networks is required in the context of project landscapes that evolve to accommodate dynamic development processes and changing requirements. On the other hand, in order to scale to today's potentially large models and allow development in teams, modular and incremental execution of these networks is required, as full re-execution of the entire network in reaction to changes may result in unacceptable execution times and loss of information \citeOwn{giese2009model}.

In this context, model queries, due to being explicitly and implicitly required by model properties and model consistency checks respectively model transformations and model synchronizations, play a central role. Solutions thus have to offer dedicated support for handling potentially complex model queries and facilitate their modular composition and reuse.

Furthermore, model operations with side-effects, such as model transformation and synchronization, and their interaction with other model operations pose a unique challenge regarding the overall goal of guaranteeing the consistency of a system description that may be distributed over multiple models. 

In this report, we propose an approach to Global Model Management that specifically aims to provide both the required modularity and incrementality. Our solution is based on an extended notion of Generalized Discrimination Networks \cite{hanson2002trigger}, a mechanism that has previously been implemented in the context of model driven engineering \citeOwn{beyhl2016operationalization} to allow a modular and incremental specification and execution of model queries in the form of nested graph conditions \cite{habel2009correctness}.

Therefore, we introduce a more general formalization called \emph{extended Generalized Discrimination Networks (\eGDN{}s)}, which (i) supports a more flexible notion of model queries, affording increased expressiveness and (ii) allows the integration of model operations with side effects into the unifying framework. In addition, we provide algorithms for the incremental execution of \eGDN{}s.

Furthermore, we integrate a number of typical model operations into a prototypical implementation of the approach and use this implementation to perform an initial evaluation of our technique's scalability using an application scenario from the software development domain. This empirical evaluation is complemented by a conceptual evaluation regarding the applicability of \eGDN{}s in a typical scenario.

The remainder of the report is structured as follows: We briefly reiterate the basic concepts of models in the form of typed graphs and discrimination networks in Chapter \ref{sec:preliminaries}. After introducing the required concepts, we discuss requirements of a solution for global model management and related work in Chapter \ref{sec:requirements}, providing further motivation for the design of a new solution. Our contribution in the form of extended Generalized Discrimination Networks is presented in Chapters \ref{sec:egdns_definition}, \ref{sec:egdns_incremental}, and \ref{sec:egdns_implementation}. Therefore, Chapter \ref{sec:egdns_definition} provides a definition of \eGDN{}s along with a graphical notation. Chapter \ref{sec:egdns_incremental} describes the incremental execution of \eGDN{}s. Chapter \ref{sec:egdns_implementation} then lists a number of examples for \eGDN{} operations that are part of our prototypical implementation. This prototypical implementation is used to perform an initial empirical evaluation of the presented concepts, which is presented in Chapter \ref{sec:evaluation} along with a conceptual evaluation of the applicability of \eGDN{}s to an example use case. Finally, Chapter \ref{sec:conclusion} concludes the report and gives an overview of possible directions for future work.


\section{Preliminaries}\label{sec:preliminaries}

In this chapter, we reiterate the basic notions of models in the form of typed graphs and discrimination networks.

\subsection{Graphs and Models}

A graph $G = (V^G, E^G, s^G, t^G)$ consists of a set of vertices $V^G$, a set of edges $E^G$, and two functions $s^G, t^G: E^G \rightarrow V^G$ assigning each edge its source respectively target vertex \cite{Ehrig+2006}. A graph morphism $m: G \rightarrow H$ between graphs $G$ and $H$ is a pair of functions $m^V : V^G \rightarrow V^H, m^E: E^G \rightarrow E^H$ such that $s^H \circ m^E = m^V \circ s^G$ and $t^H \circ m^E = m^V \circ t^G$.

A graph $G$ can be typed over a type graph $TG$ via a morphism $type^G: G \rightarrow TG$ that assigns elements from $G$ types defined in $TG$. This yields a typed graph $G^T = (G, type^G)$. A typed graph morphism $m^T : G^T \rightarrow H^T$ between two typed graphs $G^T = (G, type^G)$ and $H^T = (H, type^H)$ typed over the same type graph $TG$ is given by a graph morphism $m: G \rightarrow H$ with $type^G = type^H \circ m^T$.

In the context of this report, a model is then characterized by a typed graph, where the type graph effectively acts as a metamodel. Importantly, attributes for model elements can be realized in the framework of typed graphs by simply modeling attribute values as dedicated nodes, which leads to the notion of typed attributed graphs \cite{heckel2002confluence}. A modeling language $ML$ is defined by a graph $TG$ and denotes the set of all possible graphs typed over $TG$.

Figure \ref{fig:sample_graph} shows an example model from the software development domain in the form of a typed graph $G$, and the associated metamodel in the form of the type graph $TG$, with the typing morphism given by node labels in case of nodes and implicitly in case of edges. The example model represents the abstract syntax graph (ASG) of a program written in an object-oriented programming language. Nodes in the model represent packages, types, and methods. Edges represent containment relationships between the different concepts, with methods contained in types and types contained in packages, and return type relationships between methods and types.

\begin{figure}
\centering
\includegraphics[width=\textwidth]{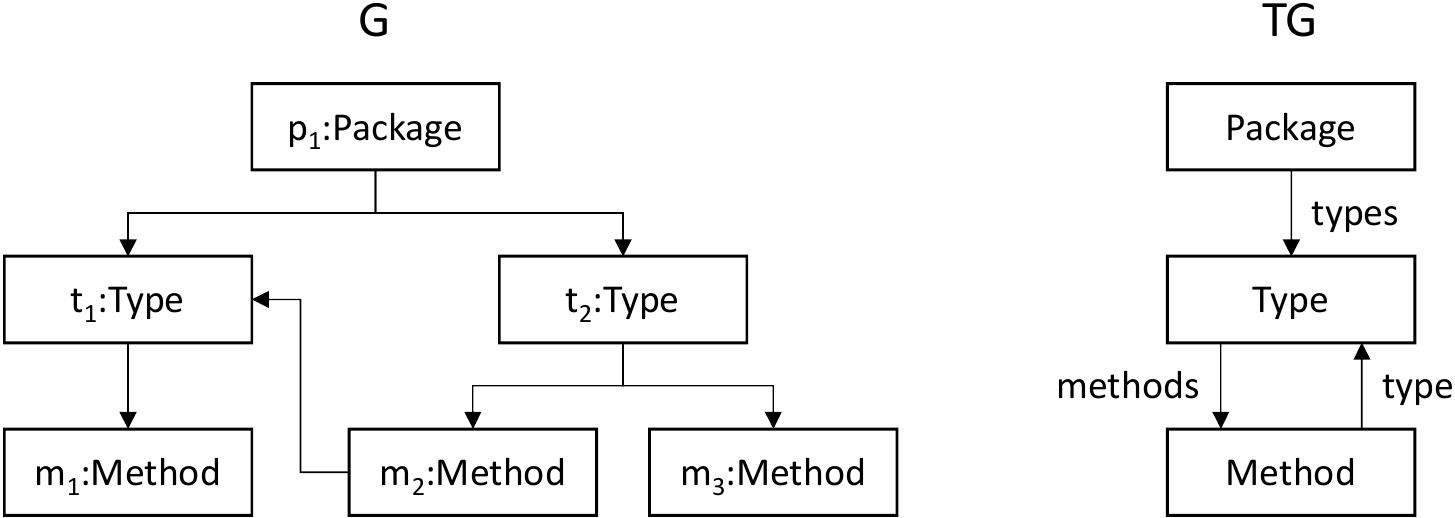}
\caption{Example model and metamodel in the form of typed graph and type graph from the software development domain} \label{fig:sample_graph}
\end{figure}

\subsection{Discrimination Networks} \label{sec:discrimination_networks}

A discrimination network is a graph of nodes representing computation units and edges representing dependencies between these units. Discrimination networks are a popular solution for the incremental execution of model queries such as the computation of model properties or the checking of model consistency conditions. Therefore, the model query is decomposed into subqueries, which form the discrimination network's nodes.

The execution of a subquery can make use of the results computed for another subquery, which is indicated by a dependency relation between the two subqueries. The execution of a final discrimination network node yields the overall query result. By storing the results of discrimination network nodes beyond the execution of a query, incremental execution that reuses previously computed results in subsequent executions is enabled.

Since discrimination networks so far are primarily employed for model querying, current approaches offer only limited or no support for the integration of model operations with side-effects and thus constitute at best a partial solution for global model management. However, due to their inherent support for modularity and incrementality, they offer a promising starting point.

There exist different realizations of the concept of discrimination networks in the context of model driven engineering, two of which will be briefly presented in the following subsections.

\subsubsection{RETE nets}

RETE nets were initially introduced by Forgy \cite{forgy1989rete} and are characterized by the fact that nodes are only allowed to have dependencies to at most two other nodes. Some examples of RETE nodes are:
\begin{itemize}
\item \emph{input nodes}, which correspond to primitive model queries that extract individual elements, that is, nodes or edges, from a model, and consequently have no dependencies
\item \emph{filter nodes}, which filter the results of some other subquery by a condition and consequently have one dependency
\item \emph{join nodes}, which combine the results of two other subqueries into results for a more complex subquery and consequently have two dependencies
\end{itemize}
While the listed node types form the core of incremental model querying solutions such as the well-established VIATRA \cite{varro2016}, RETE nets are a flexible mechanism that allows a multitude of other query-related node types. This is illustrated by VIATRA's support for various advanced constructs for specifying model queries, including negative patterns and certain aggregation operations.

In RETE implementations, results computed by a RETE net's nodes are usually stored in memory in so-called indexers, which act as implicit interfaces between computation nodes. These indexers can also be made explicit by modeling them as part of the RETE net via a different kind of RETE node that is not associated with any computational functionality, but only serves as a storage for other nodes' results.

\subsubsection{Generalized Discrimination Networks}

Generalized Discrimination Networks (GDNs) are a less restrictive form of discrimination networks than RETE nets and were developed by Hanson et al. \cite{hanson2002trigger}. Essentially, GDNs drop the limit on the number of a node's dependencies of RETE nets and thereby allow for more control over which intermediate query results are to be stored in memory.

A realization in the context of model querying was presented in \citeOwn{beyhl2016operationalization}. It implements GDN nodes as model transformation rules that create marking elements for subquery results directly as part of the queried model. Dependencies between nodes are realized by considering marking elements created by the required node in the transformation rule associated with the dependent node. However, while the approach in \citeOwn{beyhl2016operationalization} is based on a fairly expressive notion of queries in the form of nested graph conditions, certain query-related operations such as aggregation are not supported by the underlying formalism.


\section{Requirements for Global Model Management}\label{sec:requirements}

Nowadays the development of complex systems with models requires \emph{Global Model Management}
(\emph{GMM})~\cite{BJRV05,Favre04foundationsof} to ensure that the models of
different subsystems, of different views, and of different domains are properly
combined, even though the models might reside at different levels of
abstraction. Indeed, due to the heterogeneity and complexity of systems such as
Cyber-Physical Systems (CPS), it is no longer feasible to represent the system
as a Single Underying Model (SUM). This is because numerous languages and tools
are already employed independently by domain experts collaborating to build the
system. Redeveloping these tools and thus requiring industry to change its
practices is not conceivable given the required development efforts, but also
the strong resistance to change development processes. This is especially
relevant in the case of safety-critical systems that must undergo complex
certification processes. Therefore, many models must be used to represent the
system and adequate GMM is required to ensure that the development activities
that operate on the models are properly coordinated such that the models lead to
a proper system as a whole, where the different elements and aspects covered by
the different models are correctly integrated and are consistent with each
other.

A classification of model integration problems and fundamental integration
techniques has  been introduced in~\citeOwn{GieseNNS2011}. It highlights  the techniques of decomposition and enrichment, which characterize two orthogonal dimensions of development where the system is decomposed into subsystems and domains (\emph{horizontal} dimension) and into a set of models with increasing level of details (\emph{vertical} dimension). This requires coordinating all activities operating on the models across these dimensions to ensure their consistency. 
%

The development activities for nowadays complex systems are spread across multiple domains and teams, where each team is using its own set of modeling languages thus requiring proper \emph{integration} of these languages. Indeed, it has been shown that using a single language to cover all domains would lead to very large monolithic languages not easily customizable for the development environments and tools needed by development organizations. These considerations lead to \emph{Multi-Paradigm Modeling} (\emph{MPM})~\cite{VangheluweAIS2002}, which advocates the integration of reusable \emph{modular} modeling languages instead  of large monolithic languages.
Hence, GMM must support integrating with appropriate modularity not only models but also their \textit{modeling languages} (hereafter \emph{modeling language integration}), in addition to coordinating all activities
operating on the models and specified as \textit{model operations / transformations}. The execution of these model operations has to be \emph{scalable} for being able to handle large models.  This requires \emph{incrementality}, where only the  operations impacted by a model change are re-executed, thus avoiding the effort to recompute entire models as in the case of incremental code compilers.

GMM is also known as \emph{modeling-in-the-large}, which consists of establishing global relationships (e.g. model operations that generated one model from other models) between macroscopic entities (models and metamodels) while
ignoring the internal details of these entities~\cite{BJRV05}. \emph{Megamodeling}~\cite{Bezivin+2004,Favre04foundationsof} has been introduced for the purpose of describing these macroscopic entities and their relations.

Consequently, for modular and incremental global model management solutions for the modular and incremental construction and execution of I) models and modeling languages integration, II) model operations, and III) megamodels are required. 
We will outline in the following that nowadays only preliminary approaches exist that provide ad hoc
solutions for fragments of the sketched problem and that a solid understanding of the
underlying needs and challenges is currently lacking. 
In particular, the current approaches do at most offer some modularity and/or incrementality for a single aspect as modeling languages integration or model operations. However, support for handling complex modeling landscapes as a whole in a modular and incremental fashion as required for the large-scale problems that exist in practice is not offered so far. 

In the following, we will discuss the needs in more detail and review how far existing solutions that address the construction and execution of 1) models and modeling languages integration, 2) model operations, and 3) megamodels. The way the existing approaches perform along these dimensions is depicted in Table~\ref{tab:eval-approaches}, where an empty cell identifies a need that is not addressed, a \textasciitilde \ denotes partial fulfilment of the need and a + indicates that the need is  addressed sufficiently\footnote{For convenience, we use the name of the tool or project to identify an approach when it exists, otherwise the name of the first author of the publication describing the approach is used.}. This evaluation is  discussed in further details in the following sections.


\newcommand{\HC}{\cellcolor[gray]{0.8}}
\newcommand{\HCS}{\cellcolor[gray]{0.5}}
\begin{table}[!ht]
\tiny
\centering
\scalebox{0.80}{
\begin{tabular}{|p{2.5 cm}||c|c|c|c|c||c|c|c|c||c|c|c|c|}
        \hline
        \textbf{Approach} & \multicolumn{5} {c|} {\textbf{Modeling Languages Integration}} &
        \multicolumn{4} {c|} {\textbf{Model Operations}}  & \multicolumn{4} {c|}
        {\textbf{Megamodels}} \\
        \hline
        & \multicolumn{3} {c|} {\emph{Const.}} & \multicolumn{2} {c|} {\emph{Exec.}}
        & \multicolumn{2} {c|} {\emph{Const.}} & \multicolumn{2} {c|} {\emph{Exec.}} &
        \multicolumn{2} {c|} {\emph{Const.}} & \multicolumn{2} {c|} {\emph{Exec.}} \\
        \hline
        & \emph{Links} & \emph{Int.} & \emph{MMI} & \emph{Batch} & \emph{Inc.} &
        \emph{Flow} & \emph{Ctx.} & \emph{Batch} & \emph{Inc.} & \emph{Mon.} &
        \emph{Mod.} & \emph{Batch} & \emph{Inc.}
        \\
        \hline
        \hline
        \multicolumn{14} {|c|} {\emph{Modeling Languages Integration}} \\
        \hline
        \hline
        Blanc et al.~\cite{1573497} & & & & + & \HCS + & & & & & & & & \\
        \hline
        EMF IncQuery~\cite{EMF-IncQuery-website, UjhelyiSCP2015} & & & & + &
        \HCS + & & & & & & & &
         \\
        \hline
        Egyed et al.~\cite{GroherFASE2010, Egyed2006} & & & & + & \HCS + &
        & & & & & & &
         \\
        \hline
        Cabot et al.~\cite{CT06} & & & & + & \HCS+ &
        & & & & & & &
         \\
        \hline
        ACOL~\cite{LangsweirdtISORC2010} & & \HC $\sim$ & & + & & & & & &
        & & & \\
        \hline
        SmartEMF~\cite{LH09, HesselundMODELS2008, smartemf-website} & + & \HCS +
        & & & & & & & & & & & \\
        \hline
        Composite EMF Models~\cite{JurackICGT2010, compositeemfmodels-website}
        & + & \HCS + & & & & & & & & & & & \\
        \hline
        EMF Views~\cite{emfviews-website, BrunelierePWC15} & + & & \HC $\sim$ & + & & & & & &
        & & & \\
        \hline
        Kompren~\cite{ABlouinSoSyM2015, kompren-website} / Kompose~\cite{FleureyMSE2008, kompose-website} & & & \HC $\sim$ & & & & & & & & & & \\
        \hline
        Reuseware ModelSoc \cite{Johannes2010, ModelSoc-website} & & \HC $\sim$
        & \HC $\sim$ & & & & & & & & & & \\
        \hline
        Ratiu et al.~{\cite{ractiu2022reactive}} &
        + & & & \HC $\sim$ & \HC $\sim$ &
        & & & &
        & & & \\
        \hline
        K\"onig et al.~{\cite{konig2017efficient}} &
        + & & & + & &
        & & & &
        & & & \\
        \hline
        \hline
        \multicolumn{14} {|c|} {\emph{Model Operations}} \\
        \hline
        \hline
        Wires*~\cite{Wires-09, wires-website}
        & & & & & & + & & + & & & & & \\
        \hline
        ATL Flow~\cite{atlflow-website} & & & & & & + &  & + & & & & & \\
        \hline
        Epsilon~\cite{PaigeIECCS2009, epsilon-website}  & + & & & & & + & &
        + & & & & & \\
        \hline
        Gaspard2~\cite{EtienSAC2010, gaspard2-website} & & & \HC $\sim$ & &
        & + & \HC $\sim$ & + & \HC $\sim$ & & & &
        \\
        \hline
        Debreceni et al.~\cite{DebreceniFASE2014} & & & & &
        & + & \HC $\sim$ & + & \HCS + & & & &
        \\
        \hline
        \textbf{MoTCoF}~\cite{SHNG11} & & & & & & + & \HCS + & + &
        \HC $\sim$ & & & & \\
        \hline
        \textbf{MoTE}~\cite{GNH10}\cite{mote-website} & + & & & & & & & +
        & \HCS + & & & & \\
        \hline
        \hline
        \multicolumn{14} {|c|} {\emph{Integration Languages and Others}} \\
        \hline
        \hline
        CyPhy~\cite{SimkoASME2012, gme-website} & + & & & & & + & & + & & &
        & & \\
        \hline
        FUSED~\cite{BoddyAVICPS2011, fused-website} & + & \HC $\sim$ & & & & + & & + &
        & & & & \\
        \hline
        CONSYSTENT~\cite{HerzigProcCS2014, HerzigMoDeVVa2014} & + & & & & &
        + & & + & & & & & \\
        \hline
        \hline
        \multicolumn{14} {|c|} {\emph{Megamodels}} \\
        \hline
        \hline
        AM3~\cite{Vignaga_et_al:2013, am3-website} & + & & & & & + & & + & &
        + & & \HC $\sim$ & \\
        \hline
        FTG+PM~\cite{Lucio_et_al:2013, atompm-website} & & & & & & + & & + &
        & + & & + & \\
        \hline
        MegaL Exp.~\cite{Favre:2012:MLA:2404962.2404978} & + & & & & & + &
        & & & + & & & \\
        \hline
        \textbf{GMM*}~\cite{BlouinGemoc2014} & + & & \HC $\sim$ & & & + & &
        + & \HC $\sim$ & + & & + & \HC $\sim$ \\
        \hline
        \textbf{Seibel et al.}~\cite{SNG10, SHG12}\cite{Beyhl_et_al:2013} & + & & & &
        & & & & & + & & + & \HC $\sim$ \\
        \hline
        Stevens~{\cite{stevens2020maintaining,stevens2020connecting}} &
        & & & & &
        + & & + & + &
        + & & + & + \\
        \hline
        Gleitze et al.~{\cite{gleitze2021finding}} &
        & & & & &
        + & & + & + &
        + & & + & + \\
        \hline
        Vitruvius~{\cite{klare2021enabling}} &
        + & & & + & + &
        + & & + & + &
        + & & + & + \\
        \hline
        \hline
        \emph{\eGDN{}s} & \emph{+} & \HCS \emph{+} & \emph{+} & \emph{+} & \HCS \emph{+} & \emph{+} & \HCS \emph{+} & \emph{+} & \HCS \emph{+} &
        \emph{+} & \HCS \emph{+} & + & \HCS \emph{+} \\
        \hline
\end{tabular}
}
\caption{Comparison of existing and planned global model management approaches}
\label{tab:eval-approaches}
\vspace{-5mm}
\end{table}

\subsection{Models and Modeling Languages Integration: Construction and Execution}%
\subsubsection{Construction}
The construction of models and modeling languages integration is addressed in the current approaches in three
main ways  via (1) linking of models and model elements, (2) model interfaces
and (3) metamodel composition.

\textbf{(1) Links:~ }
 
All approaches make use of some kind of \textit{trace links\textit{ }}between
models and their model elements to integrate models.
In this report, we adopt the definitions of traceability
proposed by the Center of Excellence for Software Traceability (CoEST)~\cite{CoEST-website}. A \emph{trace link} is \emph{"...a specified association between a pair of artifacts, one comprising the source artifact and one comprising the target artifact..."}. Following the CoEST again, trace links are specialized into traces between the \emph{vertical}
and \emph{horizontal} dimensions. Hence, a \emph{vertical} trace
\emph{"...links} \emph{artifacts at different levels of abstraction so as to accommodate lifecycle-wide or end-to-end traceability, such as from requirements to code..."}. An horizontal trace links   \emph{"...artifacts at the same level of abstraction, such as: (i) traces between all the requirements created by ‘Mary’, (ii) traces between requirements that are concerned with the performance of the system, or (iii) traces between versions of a particular requirement at different moments in time"}.

There is a plethora of approaches (e.g.,
\cite{amw-website,epsilon-website,Favre:2012:MLA:2404962.2404978,LH09,HerzigProcCS2014,SimkoASME2012,BoddyAVICPS2011}~\citeOwn{mote-website}) making use of trace links to integrate models. The Atlas Model Weaving (AMW) language~\cite{amw-website} provided one of the first approaches for capturing hierarchical traceability links between models and model elements.
The purpose was to support activities such as automated navigation between
elements of the linked models. In this approach, a generic core traceability language is made available and optionally extended to provide semantics specific to the metamodels of the models to be linked. Similarly, the Epsilon framework~\cite{epsilon-website} provides a tool (ModeLink) to establish correspondences between models. MegaL Explorer~\cite{Favre:2012:MLA:2404962.2404978} supports  relating heterogeneous  software development artifacts which do not necessary have to be models or model elements using predefined relation types. 
SmartEMF~\cite{LH09} is another tool for linking models based on annotations of Ecore metamodels to specify simple relations between model elements through correspondence rules for attribute values. Complex relations are specified with ontologies relating the concepts of the linked languages. The whole set of combined models is converted into Prolog facts to support various activities such as navigation, consistency and user guidance when editing models. The CONSYSTENT tool and approach~\cite{HerzigProcCS2014} make use of a similar idea. However, graph structures and pattern matching are used to represent the  combined models in a common formalism and to identify and manage inconsistencies instead of Prolog facts as in the case of SmartEMF.

There are also a number of approaches such as~\cite{SimkoASME2012} and
\cite{BoddyAVICPS2011} that build on establishing links between models through the use of integration languages developed for a specific set of integrated modeling languages, where the integration language embeds constructs specific to the linked languages. This is also the case for model weaving languages extending the core AMW language. However, AMW has the advantage of capturing the linking domain with a core common language.
Other means for linking and integrating models are Triple Graph Grammars (TGG) such as the Model Transformation Engine (\textbf{MoTE}) tool~\citeOwn{mote-website}, which  similarly requires the specification of some sort of integration language (correspondence metamodel) specific to the integrated languages. However, an important asset of this approach is that it automatically establishes and manages the traceability links and maintains the consistency of the linked models (model synchronization) in a \emph{scalable}, \emph{incremental} manner.
Finally, in~\citeOwn{SNG10,SHG12}\citeOwn{Beyhl_et_al:2013}, an approach is presented to automatically create and maintain traceability links between models in a scalable manner. While the approach focuses on traceability management rather than model integration, compared to integration languages, it relies on link types defined at the model level (and not at the metamodel / language level), thus avoiding the need to update the integration language every time a new language must be integrated.                  
More recently, the concept of reactive links has been presented in \cite{ractiu2022reactive}, which essentially allows an incremental propagation of attribute value changes between models of different languages. However, incremental execution is only offered for a limited notion of consistency.

The comparison of these approaches shows that apart from our own earlier approach~\citeOwn{SNG10,SHG12}\citeOwn{Beyhl_et_al:2013}, all approaches suffer from being dependent on the set of integrated languages, thus requiring to better support modularity. Furthermore, only our own work~\citeOwn{mote-website}\citeOwn{SNG10,SHG12}\citeOwn{Beyhl_et_al:2013} supports automated management of traceability links.

\textbf{(2) Interfaces:~ }
In addition to links, a few more sophisticated approaches
(e.g.,~\cite{LangsweirdtISORC2010, HesselundMODELS2008, JurackICGT2010,
Johannes2010}) introduce a concept of \emph{model interface} (\emph{int.} column
in Table \ref{tab:eval-approaches}) for specifying how models can be linked.
In~\cite{LangsweirdtISORC2010}, the Analysis Constraints Optimization Language
(ACOL) is proposed, which has been designed to be pluggable to an Architecture
Description Language (ADL).
A concept of \emph{interface} specific to ACOL is included so that constraints
can refer to these interfaces to relate to the model elements expected from the
ADL. SmartEMF~\cite{HesselundMODELS2008, smartemf-website} proposes a more
generic concept of model interface to track dependencies between models and metamodels and provide automated compatibility checks.
Composite EMF Models \cite{JurackICGT2010, compositeemfmodels-website}
introduces \emph{export} and \emph{import} interfaces to specify which model
elements of a main model (\emph{body}) should be exposed to other models (i.e.
are part of the public API), and which elements of a body model are to be
required from an export interface. In~\cite{Johannes2010}, an approach
for the composition of grammars with explicit variation points (hooks)
constituting an implicit \emph{invasive} composition interface is presented.

However, while these approaches provide interesting preliminary ideas, they
need to be enriched to cover a larger number of non intrusive model integration
use cases such as for example, specifying modification policies of the linked model elements required to ensure the models
can be kept consistent. They also lack integration into GMM.

\textbf{(3) Metamodel Integration:~ } 
Some approaches (e.g.,~\cite{kompren-website, kompose-website,
EtienSAC2010,emfviews-website}~\cite{BlouinGemoc2014}) consider the
construction of view metamodels in terms of other metamodels or language
fragments (\emph{MMI} column in Table \ref{tab:eval-approaches}). 
In~\cite{EtienSAC2010}, an approach implemented in the Gaspard2 tool
~\cite{gaspard2-website} is presented where metamodels are artificially
extended for the purpose of combining independent  model transformations
resulting in an extended transformation for the extended metamodels.
In~\cite{ABlouinSoSyM2015}, a language  and tool
(Kompren)~\cite{kompren-website} are proposed to specify and generate slices of
metamodels via the selection of classes and properties of an input metamodel.
A reduced metamodel is then produced, which must be completely regenerated when the input metamodel is
changed.
Such is the case for the Kompose approach \cite{kompose-website}, which on the
contrary to Kompren proposes to create \emph{compound metamodels,} where a set
of visible model elements from each combined metamodels is selected and
optionally related. EMF Views~\cite{emfviews-website, BrunelierePWC15} provides
similar approach however without the need to duplicate the metamodel elements
as opposed to Kompose and Kompren. Indeed, EMF Views allows the specification of \emph{virtual} metamodels that only refer to existing metamodel elements instead of duplicating them. The same principle applies for the  given models of the virtual metamodels, which only refer to elements of the existing integrated models instead of duplicating them. The defined
virtual view metamodels are usable transparently by tools. Furthermore, the same models can  be simultaneously used by both legacy tools and new tools making use of the virtual metamodels, thanks to the non-intrusiveness of the approach.
Finally, the Global Model Management language (\textbf{GMM*})\footnote{We use *
to distinguish this existing language and tool from the generic Global Model
Management (GMM) acronym.} \cite{BlouinGemoc2014} provides means to specify
and interpret reusable language subsets as sets of constraints combined to form
subsetted metamodels. Like for EMF Views, these reduced metamodels can to some
extent be used transparently by tools.
Aspect-oriented metamodel composition
is another well-known technique for metamodel composition. However it requires metamodels to be expressed in a specific aspect-related format, which does not meet our non-intrusiveness requirement. 

While each of these approaches provides interesting support for modular modeling
languages integration, their unification into a common formalism, the use of an explicit
notion of a model interface and their integration into GMM is lacking, except
for subsetted metamodels already integrated within our \textbf{GMM*} language.
Among these approaches, we note that EMF Views provides an adequate
starting point for this work, due to its \emph{non-intrusiveness} property essential for reusing legacy models and tools. However, in its current
implementation, only changes of attributes of virtual compound models
are propagated to the underlying real models ~\cite{BrunelierePWC15}. Other
changes propagation as well as metamodel constraints composition remain to be
addressed. The integration of an explicit metamodel interface construct
for governing how metamodels can be composed, as well as the ability to solve
attribute and operation conflicts of merged classes inspired from the concept of
\emph{Traits} / \emph{Mixins} developed for object oriented programming are required future
works for this approach.
%

%
Execution of integrated models concerns the evaluation of the
well-formedness constraints of each combined model alone, but also of the combined models as a whole. To our knowledge, no approach addresses the incremental checking of well-formedness conditions across the different language fragments of compound models. However, some approaches on incremental
constraints evaluation exist. In~\cite{1573497}, changes on models are expressed as sequences of atomic model operations to determine which constraint is impacted by the changes, so that only these constraints need to be re-evaluated. In~\cite{EMF-IncQuery-website, UjhelyiSCP2015},  a graph-based query language (EMF-IncQuery) relying on incremental pattern
matching for improved performance is also proposed. In~\cite{Egyed2006}, an
approach is presented for incremental evaluation of constraints based on a scope of model elements referenced by the query and determined during the first query evaluation. This scope is stored into cache and used to determine which queries need to be re-evaluated according for some model changes.
In~\cite{GroherFASE2010}, this approach is extended for the case where the constraints themselves may change besides the 
constrained models. Finally in~\cite{CT06}, an incremental OCL checker is presented where a simpler OCL expression and reduced context elements set are computed from an OCL constraint and a given structural change event. Evaluating this simpler constraint for the reduced context is sufficient to assert the validity of the initial constraint and requires significantly less computation resources.

In \cite{konig2017efficient}, K\"onig et al. introduce a technique for the checking of consistency constraints over linked models, which avoids the merging of these models into a single underlying model to achieve better scalability. However, while formally defined and proven to be correct, the approach in \cite{konig2017efficient} does not consider incremental consistency checking.

We identified the following requirements as main needs concerning modularity and incrementality of modeling languages integration:

\begin{itemize}
\item[] \requirement{1.1} modeling languages integration via integration links and combination of well-formedness conditions with consistency
\item[] \requirement{1.2} interfaces for embedding of modeling languages
\end{itemize}

Note that concerning Table \ref{tab:eval-approaches} the requirements cover here Links and Interfaces which jointly emulate the less modular direct meta model integration and that  the employed well-formedness conditions and consistency conditions will be covered when we consider model operations in the next section.
Consequently, as visible in Table \ref{tab:eval-approaches}, there yet does not exists any approach that provides a combination of all these requirements we target.

\subsection{Model Operations: Construction and Execution}%
The construction of model operations is addressed in two ways in the literature. Most approaches combine model operations as \emph{model transformations chains} ((1) \emph{Flow Composition}), where each chained transformation operates at the granularity of complete models. In order to support reuse and scalability for complex modeling languages, which are defined by composing them from simpler modeling languages, a few approaches have considered specifying model transformations as white boxes. Composed of explicit fine grained  operations processing model elements for a given context, these operations are reusable across several model transformations ((2) \emph{Context Composition}).

\subsubsection*{(1) Flow Composition Approaches:~ }
FUSED (Formal United System Engineering Development)~\cite{BoddyAVICPS2011} is
an integration language to specify complex relationships between models of different languages. It supports model transformation chains, but only implicitly via execution of tools, without explicit representation of the involved transformations and processed data. On the contrary, there is a plethora of approaches allowing the explicit specification and construction of model transformation chains implementing a data flow paradigm. Such is the case of the AtlanMod Megamodel Management (AM3) tool~\cite{am3-website}, for which the Atlas Transformation Language (ATL)~\cite{atl-website} is used to specify the model transformations. Besides, a
type system has been developed~\cite{Vignaga_et_al:2013}, which enables type checking and inference on artifacts
related via model transformations.
Another similar but less advanced tool is the Epsilon Framework~\cite{epsilon-website},
which  provides model transformation chaining via ANT tasks. Wires~\cite{Wires-09} and ATL Flow~\cite{atlflow-website}
are tools providing graphical languages
for the orchestration of ATL model transformations.  The Formalism Transformation Graph + Process Model (FTG+PM) formalism~\cite{Lucio_et_al:2013} implemented in the AToMPM (A Tool for Multi-Paradigm Modeling) tool~\cite{atompm-website} provides similar functionality. However, it has the advantage of also specifying the complete modeling process in addition to the involved model transformations. This is achieved via activity diagrams coupled with model transformation specifications executed automatically to support the development process. Finally, \textbf{GMM*} {\cite{BlouinGemoc2014}}  also supports model transformation chaining, but through the specification of relations
between models of specific metamodels that can be chained. One advantage of this approach is that automated incremental (re-)execution of the specified relations between models is provided in response to received model change events.
Incrementality of the execution of the transformations is
also made possible by the integration of the  
\textbf{MoTE}~\citeOwn{mote-website} incremental model transformation tool into \textbf{GMM*}. 

However, while chaining model transformations offers some degree of modularity of model transformation specifications, apart from \textbf{GMM*}, most approaches suffer from scalability issues for large models, since the used transformation tools do not support incremental execution. In addition, the case where a generated model is modified by hand to add  information not expressible with the language of the original model(s) cannot easily be handled by these approaches, since regenerating the model modified by hand will destroy the user-specific information. This need is better supported by context composition approaches.
 
\subsubsection*{(2) Context Composition Approaches:~ }
A few approaches allow context composition of model operations (column \emph{Ctx.} in Table \ref{tab:eval-approaches}). In~\cite{EtienSAC2010} as mentioned above, an approach is described to combine independent model transformations resulting in extended transformations for corresponding extended metamodels. In~\cite{DebreceniFASE2014}, an approach is described for specifying the construction of view models using contextual composition of model operations (derivation rules) encoded as annotations of queries  of the EMF IncQuery~\cite{EMF-IncQuery-website} language.  Traceability links between view and source model elements are automatically established and maintained. The use of EMF IncQuery natively provides incremental execution of the derivation rules to synchronize the view model with the source model. Some views may be derived from other
views thus allowing flow composition as chains of view models. This approach achieves results similar to TGGs supporting incrementality, however with the drawback of being unidirectional. Similarly, but with bi-directionality the \textbf{MoTCoF} language~\citeOwn{SHNG11}
allows for both flow and fine grained context composition of model transformations. An advantage over~\cite{EtienSAC2010} however is that model transformations are used as black boxes without the need to adapt the transformations according to the context.

As can be seen, most approaches only support flow type modularity for model operations with batch execution except for our \textbf{GMM*} language thanks to its integration of \textbf{MoTE} providing incremental execution. This will not scale and lead to information losses in case of
partial model information overlap. Only a few approaches allow context modularity, which better supports incremental application where only the impacted operations can be re-applied following a change in order to avoid the cost of re-computing complete transformations.
Such is the case of \textbf{MoTCoF}, which theoretically permits incremental execution, but a concrete technical solution is still lacking for it.



	 

%
%


To address modularity and incrementality for model operations, we identified as main needs:

\begin{itemize}
\item[] \requirement{2.1} composition of model operations
\item[] \requirement{2.2} model operations over integrated models
\item[] \requirement{2.3} execution scheme for model operations
\end{itemize}

Note that concerning Table \ref{tab:eval-approaches} the requirements cover here Flow and Context based composition and Batch as well as Incremental Execution at first for all special cases of model operations and then also for the general case.
Consequently, as visible in Table \ref{tab:eval-approaches}, there yet does not exists any approach that fully cover the envisioned combination of all these requirements we target.

\subsection{Megamodels and other Global Model Management Approaches}%
Two strands can be identified for GMM. A first one makes use of
\emph{(1) model integration languages,} which are  defined for a specific set of
integrated modeling languages and tools meaning that the integration language must be updated every time a new language or tool is used. The second strand attempts to solve this problem by making use of \emph{(2) megamodels} providing configurable global model management. 

\subsubsection*{(1) Integration Languages and other Approaches:~ }
The CyPhy~\cite{SimkoASME2012} used in the GME modeling tool~\cite{gme-website} and FUSED~\cite{BoddyAVICPS2011,fused-website} are examples of model integration languages. But as mentioned above, these languages must be adapted as soon as a different set of integrated languages and tools must be used, thus requiring highly skilled developers. Integration languages are therefore not practical.

Open Services for Lifecycle Collaboration (OSLC)~\cite{OSLC-website} provides standards for tool integration through the Web. Many specifications are available for \emph{change management}, \emph{resource previews}, \emph{linked data}, etc. It builds on the W3C \textit{linked data} standard, which aims at providing best practices for publishing structured data on the Web based on the W3C Resource  Description Framework (RDF). RDF is a model for data interchange on the Web where data is represented as graphs. However, OSLC is  more services (and tools) oriented and inherits the problems of \emph{linked data}, which is specific to the Web and therefore does not separate the concerns of data representation and persistence as opposed to Model-Driven Engineering (MDE) where an abstract syntax is used independently of the way the data is stored.

Another approach making use of these standards is \cite{HerzigProcCS2014} and is implemented in a tool named CONSYSTENT, used to identify
and resolve inconsistencies across viewpoints due to information overlapping. The information of all
models involved during development is captured in a common RDF graph. The approach relies on a human\footnote{An automated method making use of Bayesian Belief Networks is also under study~\cite{HerzigMoDeVVa2014}.} to
specify patterns representing semantic equivalence links (semantic connections) across the graph models.
Inconsistency patterns based on these semantic connections are continuously checked
over the RDF model for potential matches identifying inconsistencies. Means to automatically resolve inconsistencies are under development.
However, this approach necessitating the conversion of all models as a RDF graph is not incremental and will not scale for
large models.

\subsubsection*{(2) Megamodels:~ }
In this second strand, megamodels serve to capture and manage
MDE resources such as modeling languages, model transformations, model correspondences, and tools used in modeling environments.
There are several megamodeling approaches as already mentioned. AM3 \cite{am3-website} is one of the first ones where a megamodel is basically a registry for MDE resources. Model transformations are specified with
ATL~\cite{atl-website} and model correspondences with the Atlas Model Weaving (AMW)
language [2]. Similarly, FTG+PM~\cite{Lucio_et_al:2013} as mentioned above is also a megamodeling language as well as MegaL Explorer ~\cite{Favre:2012:MLA:2404962.2404978} allowing to model the artifacts used in
software development environments and their relations from a linguistic point of view. The involved software
languages and related technologies and technological spaces can be captured with linguistic relationships
between them such as membership, subset, conformance, input, dependency, definition, etc. Operations
between entities can also be captured. The artifacts do not need to be represented as models, but each
entity of the megamodel can be linked to a Web resource that can be browsed and examined.
However, the language seems to be used mostly for visualization providing a better understanding of the
developments artifacts but cannot be executed to perform model management.
The aforementioned \textbf{GMM*} infrastructure~\cite{BlouinGemoc2014} consists of a megamodeling language inspired from~\citeOwn{HSG2010}. Metamodels can be declared, as well as relations between models of these metamodels. In
particular, synchronization relations can relate models of two different metamodels making use
of the \textbf{MoTE} TGG engine~\citeOwn{mote-website} to transform or synchronize the models. As mentioned earlier, chains of model transformations
can be specified
and executed incrementally in response to model change events and \emph{subsets} of modeling languages can be declared. \textbf{GMM*} is experimented
within the Kaolin tool {\cite{BlouinAHS2015}} making use of complex and rich industrial languages such as AADL and VHDL thus challenging GMM for realistic specifications.


A new approach to modeling in the large with bidirectional model transformation has been proposed by Stevens \cite{stevens2020maintaining}. The work in \cite{stevens2020maintaining} presents a formalized notion of a megamodel in the form of a hypergraph, where models are represented as nodes that can be connected via hyperedges representing bidirectional transformations. Incremental execution is generally supported by the formalism, however, a concrete algorithm is only presented for megamodels with a restricted structure for which a certain notion of correctness can be guaranteed. The author extends her work in \cite{stevens2020connecting} by connecting her previous work to research in the domain of build systems and introducing a so-called orientation model to steer megamodel execution, relaxing the restrictions on the megamodel's structure while maintaining a formal guarantee of correctness. However, the construction of an orientation model is a manual and potentially challenging process for complex networks of model operations. Furthermore, the work in \cite{stevens2020maintaining,stevens2020connecting} abstracts from the technical realization of model operations and hence does not explicitly consider how operations such as the computation of model properties may be composed in a modular manner.

In \cite{gleitze2021finding}, Gleitze et al. propose an incremental execution strategy for networks of model transformations, specifically aiming for a solution that provides explanations of cases where the strategy failed to produce a consistent result. While their strategy is applicable to networks with arbitrary structure, only bidirectional transformations between pairs of models are considered, limiting the notion of supported model operations.

Recently, significant progress has also been made in the field of model views \cite{bruneliere2019feature}, which studies how consistent view models can be derived from a system description consisting of multiple interrelated models and therefore also relates to Global Model Management. The most comprehensive and advanced model view technique is probably the Vitruvius approach \cite{klare2021enabling}, which relies on a so-called virtual single underlying model (V-SUM) for the description of the overall system under development. The V-SUM is used to integrate the individual models describing system parts and derive new view models via consistency relations. Therefore, Vitruvius employs a dedicated incremental algorithm for executing complex networks of consistency preservation operations. However, the notion of consistency in \cite{klare2021enabling} is limited to relations between pairs of tuples of model elements and hence does not support certain model operations such as computation of model properties using aggregations. Furthermore, intra-model well-formedness is deliberately not covered and reuse at the mega-model level is not considered.

However, most of these megamodeling approaches only cover to a certain degree the core ingredients of specifying MDE resources by means of metamodels and model
operations with appropriate modularity and incrementality. Only fragments of the problem are solved. Furthermore, all these megamodeling languages are monolithic (column \emph{Mon.} in Table \ref{tab:eval-approaches}) and as a result, predefined megamodel fragments cannot be composed and reused to avoid rebuilding complete megamodel specifications from scratch  for new projects. We note however that aspect-oriented metamodel composition may be used as an inspiring point and adapted to megamodeling for the specification of distributed megamodels fragments contributing cross-cutting information in an integrated  megamodel. As for megamodel execution, FTG+PM, \cite{stevens2020maintaining,gleitze2021finding,klare2021enabling}, \textbf{GMM*}, and ~\citeOwn{SNG10, SHG12} consider automated or semi-automated execution in response to model changes or modeling events from the tool's user interfaces.




%

The related work demonstrates that for global model management, we need a view that combines all its facets in a mega model. 
To address modularity and incrementiality for modamodels we can conclude that the main needs are:

\begin{itemize}
\item[] \requirement{3.1} a megamodeling language with
\item[] \requirement{3.1.1} support for metamodels, well-formedness, model operations, integration views, and traceability links
\item[] \requirement{3.1.2} a megamodel operation module concept
\item[] \requirement{3.2} a robust incremental megamodel execution scheme
\item[] \requirement{3.3} megamodel interfaces
\item[] \requirement{3.4} an asynchronous incremental megamodel execution scheme
\end{itemize}

Note that concerning Table \ref{tab:eval-approaches}, the requirements cover here the modular construction as well as incremental execution.
As visible in Table \ref{tab:eval-approaches} there do exist three approaches that do not support modularity but provide a combination of all the other requirements we target. 
However, neither of them provides the required robust incremental megamodel operation execution scheme. The technique in \cite{stevens2020maintaining,stevens2020connecting}, while providing formal guarantees regarding correctness and termination, is limited to networks of model operations in the form of trees of synchronizations between pairs of models or requires the manual construction of an orientation model. The Vitruvius approach \cite{klare2021enabling}, by virtue of employing a fixpoint iteration, does not introduce any restrictions regarding the network's structure, but consequently does not guarantee termination. The execution scheme presented in \cite{gleitze2021finding} is applicable to networks of model synchronizations between pairs of models with arbitrary structure and also guarantees termination. However, outside of performing the actual execution on the concrete instance, it provides no means of determining whether a network will eventually terminate with a correct result.

\subsection{Summary of the state of the art}\label{sec:state-of-the-art:summary}%

This survey of the state of the art demonstrates that several approaches address the needs for modularity and incrementality raised in this report. However, none of them fulfill these needs at the three levels of model operations, modeling languages integration  and megamodels that we identify as being required all at once. 
Moreover, for certain individual aspects of Global Model Management, solutions with adequate modularity and incrementality do not even exists yet on their own.
This work specifically targets these essential needs that have not been sufficiently addressed yet.


\section{Extended Generalized Discrimination\\Networks}\label{sec:egdns_definition}

In this chapter, we introduce a notion of extended Generalized Discrimination Networks (\eGDN{}s) and explain how the new formalism can be used as a language for megamodels.

\subsection{Definition of \eGDN{}s}


In order to address shortcomings of current solutions and enable the modular and incremental construction and execution of complex nets of model operations such as model properties, model consistency operations, model transformations, and model synchronization, we further generalize the idea of Generalized Discrimination Networks \citeOwn{beyhl2016operationalization} to extended Generalized Discrimination Networks. Therefore, we introduce a generalized notion of GDN nodes and their interfaces. This enables the integration of model operations with side-effects and allows a more flexible definition of queries in comparison to \citeOwn{beyhl2016operationalization}, which also affords increased expressiveness.

An \eGDN{} $G = (O, S, E, s, t)$ is essentially a bipartite graph with two kinds of nodes, slot nodes and operation nodes, where $O$ is the set of operation nodes and $S$ is the set of slot nodes. Operation nodes can be connected to slot nodes and vice-versa via edges from the set of edges $E$. The source and target functions of edges are given by $s : E \rightarrow O \cup S$ respectively $t: E \rightarrow O \cup S$.

Operation nodes represent model operations or building blocks thereof, that is, suboperations. Slot nodes store information used by model operations and their suboperations in the \eGDN{}. Edges represent dependency relationships between operation and slot nodes, with the source of an edge representing the required node and the target of the edge representing the dependent node. An operation node depending on a slot nodes indicates that the corresponding model operation uses information stored in that slot. A slot node having a dependency on an operation node means that the operation node's model operation modifies the slot's contents.

We denote the set of dependencies of a slot or operation node $n$ in $O \cup S$ by $in(n) = \{d \in O \cup S | \exists e \in E: s(e) = d \wedge t(e) = n\}$. Similarly, we denote the set of dependent nodes of $n$ by $out(n) = \{d \in O \cup S | \exists e \in E: s(e) = n \wedge t(e) = d\}$. $G$ is bipartite in the sense that $\forall o \in O: in(o) \subseteq S \wedge out(o) \subseteq S$ and $\forall s \in S : in(s) \subseteq O \wedge out(s) \subseteq O$. For an operation node $o \in O$, we also refer to the set of slot nodes $in(o)$ as the \emph{input slots} of $o$ and to the set of slot nodes $out(o)$ as the \emph{output slots} of $o$.

A slot node $s$ is always associated with a modeling language $ML$ or an ordered set of variables $var = \{v_1, v_2, ..., v_k\}$ and contains a model (typed graph) of $ML$ respectively a set of variable assignments for $var$. A variable assignment for an ordered set of variables $var = \{v_1, v_2, ..., v_k\}$ is given by a tuple in $dom_V(v_1) \times ... \times dom_V(v_k)$, where $dom(v_i)$ denotes the domain of variable $v_i$, which can either be a set of nodes or edges from one or more models or a set of primitives, e.g. $\mathbb{N}$. We refer to the set of possible contained assignment sets or models of $s$ as the slot's domain, which is given by $dom(s) = ML$ in case $s$ is associated with a modeling language $ML$ or by $dom(s) = \powset{dom_V(v_1) \times ... \times dom_V(v_k)}$ if $s$ is associated with an ordered set of variables $var = \{v_1, v_2, ..., v_k\}$. Contents are then assigned to an \eGDN{}'s slots via a valuation function $val : S \rightarrow \bigcup_{s \in S} dom(s)$, such that $\forall s \in S : val(s) \in dom(s)$.

In addition to regular models, we also allow model slots to contain \emph{linking models}. The only difference between a regular model and linking model is the fact that a linking model's set of vertices may reference vertices from other regular and linking models as edge targets, thus allowing the establishment of inter-model connections. Therefore, similarly to linking models, the metamodel of a linking model, that is, the type graph of a linking model, may refer to vertices from other type graphs as edge targets.

Regarding operation nodes, we further distinguish between \emph{query nodes}, \emph{transformation nodes}, and \emph{mixed nodes}.

Query nodes extract information from models and/or other queries' results. Therefore, a query node $q$ may have an arbitrary number of input slots and exactly one output slot. $q$'s input slots may contain both models or sets of variable assignments, whereas $q'$s output slot may only contain a set of variable assignments.

Transformation nodes create or modify models based on models and/or query results. Therefore, a transformation node $t$ may have an arbitrary number of input and output slots. $t$'s input slots may contain both models or sets of variable assignments, whereas $t$'s output slots may only contain models.

A mixed node $x$ constitutes a combination of query and transformation nodes and may have an arbitrary number of input and output slots, which may contain both models or sets of variable assignments.

Each operation node $o$ with input slots $in(o) = \{s_{i_1}, ..., s_{i_k}\}$ is associated with a semantics function $\gamma_S: dom(s_{i_1}) \times ... \times dom(s_{i_k}) \rightarrow \powset{\mathbb{F}}$, where $\mathbb{F}$ denotes the set of functions $f: out(o) \rightarrow \bigcup_{s_o \in out(o)}dom(s_{o})$ such that $\forall s_{o} \in out(o) : f(s_{o}) \in dom(s_{o})$. Essentially, the semantics function of an operation node describes a consistency relationship between the operation's input and output slots.

To indicate that the contents of the slots adjacent to $o$ are consistent with $o$'s semantics function for a valuation function $val$, we write $o.valid(val)$. Formally, $o.valid(val) \leftrightarrow \exists f \in \gamma_S(val(s_{i_1}), ..., val(s_{i_k})) : \forall s_o \in out(o) : f(s_o) = val(s_o)$.

A valuation function $val$ for an \eGDN{} $G = (O, S, E, s, t)$ is consistent with $G$ as a whole if it holds that $\forall o \in O : o.valid(val)$.

\subsection{\eGDN{}s as Megamodels}

Since an \eGDN{} encodes a network of model operations connecting a set of potentially integrated models, it represents a \textbf{megamodel}. The definition of \eGDN{}s thus constitutes a language for megamodels.

Importantly, \eGDN{}s allow the composition of model operations from nodes that realize suboperations. In addition, they also allow hierarchical composition: An \eGDN{} (and therefore also a basic GDN or RETE net) can be interpreted as an \eGDN{} operation node. The input and output slots are given by the input respectively output slots of its nodes that are connected to another operation node of the parent \eGDN{}. Any slots of the child \eGDN{} without a connection to another node of the parent \eGDN{} can act as internal slots of the child and do not have to be exposed to the parent. However, some such potential internal slots may also be considered input or output slots if their contents are relevant to human users. The semantics function of an operation node representing a sub-\eGDN{} is then implicitly defined by the semantics functions of that \eGDN{}'s own operation nodes.

In addition to (hierarchical) composability, \eGDN{}s support modularity in the sense that the semantics of an operation node regarding its output slots directly depend only on the contents of its immediate input slots. Thereby, integrating additional operation nodes (along with additional slot nodes) into an \eGDN{} only requires appropriate wiring with the node's input and output slots, but is completely independent of any other operation nodes. Effectively, slots thus act as interfaces between model operations.

\eGDN{}s can also enable modularity at the model level by using the results of query nodes, potentially along with transformation nodes for propagating changes from the query results back to the base model, as model interfaces or views. For instance, simple projection queries in combination with access restrictions can be employed to implement different visibilities for different roles in a development process. Alternatively, dedicated view models in conjunction with bidirectional model synchronization operations can similarly serve to implement editable model views in the context of \eGDN{}s.


Figure \ref{fig:egdn_notation} shows our graphical notation for the visualization of \eGDN{}s. Slot nodes are depicted as rectangles and labelled ``A'' in the top right corner in the case of assignment slots and ``M'' in the case of model slots. Model slots that contain linking models are connected to the model slots containing the linked models via dashed arrows for visual clarity. Operation nodes are visualized as rectangles with rounded corners, with query nodes such as model properties or model consistency checks labelled ``Q'', transformation nodes such as model transformations and model synchronizations labelled ``T'', and mixed nodes such as sub-\eGDN{}s labelled ``X'' in the top right corner. In addition, all nodes are labeled according to the schema \emph{<name>:<type>}.

\begin{figure}
\includegraphics[width=\textwidth]{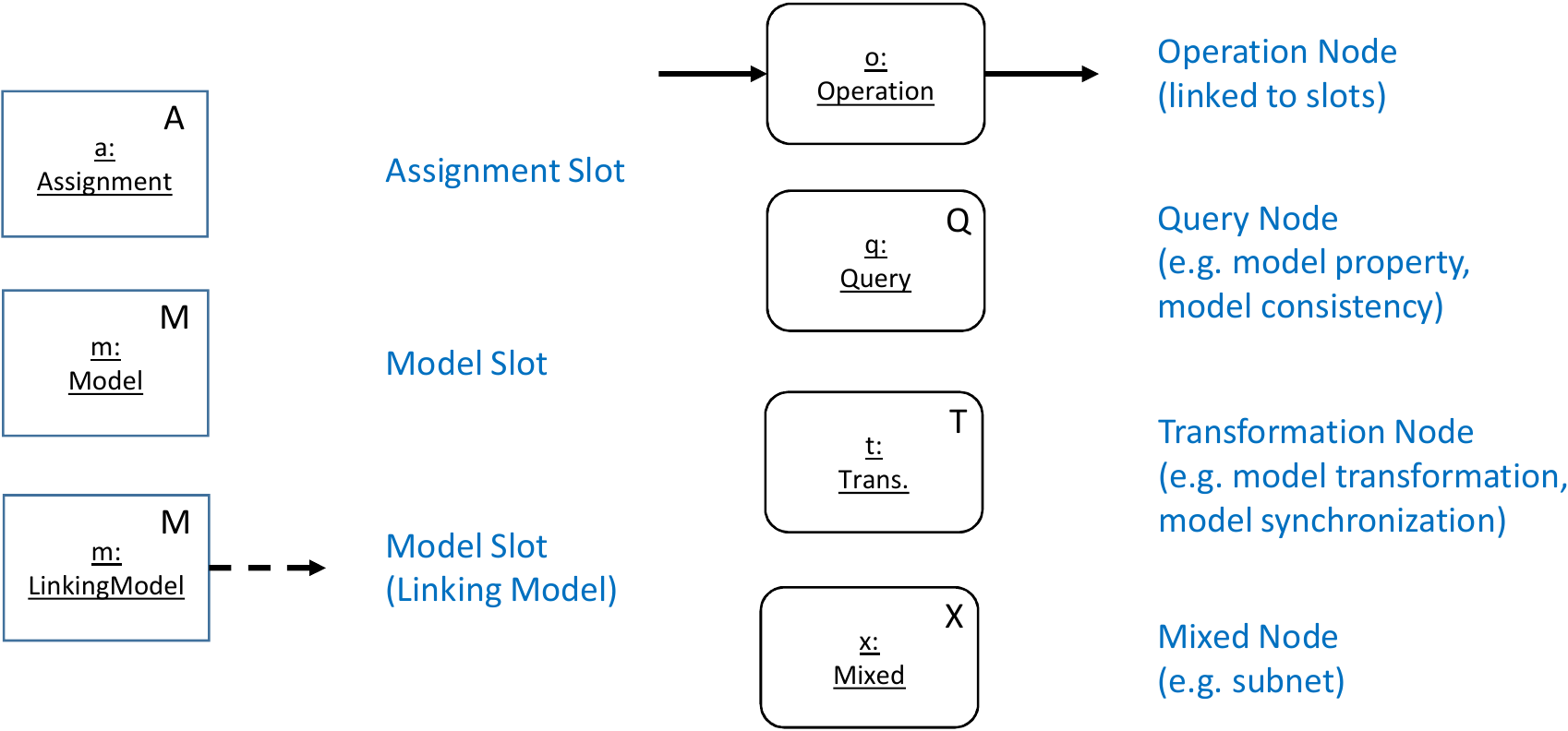}
\caption{Graphical notation for \eGDN{}s} \label{fig:egdn_notation}
\end{figure}


Figure \ref{fig:sample_egdn_simple} shows an example \eGDN{} that consists of three slot nodes and two operation nodes and realizes a simple chain of model operations. A class diagram stored in the leftmost slot is transformed into an abstract syntax graph via a transformation node. Then, a query node extracts some information from the created abstract syntax graph and makes the query result accessible via an assignment slot.

\begin{figure}
\includegraphics[width=\textwidth]{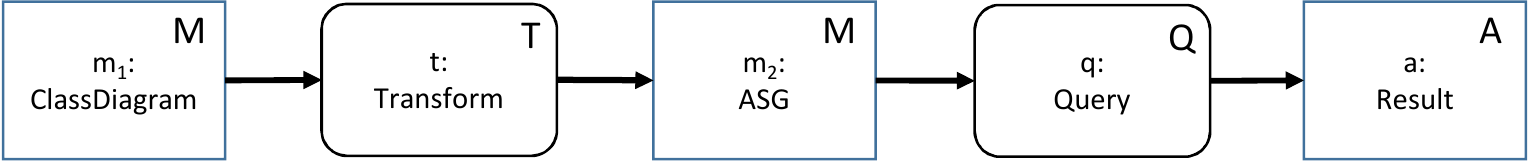}
\caption{Simple example \eGDN{}} \label{fig:sample_egdn_simple}
\end{figure}


\section{Incremental Execution of Extended Generalized Discrimination Networks}\label{sec:egdns_incremental}

In this chapter, we describe how \eGDN{}s can be executed to restore consistency in a network of models and model operations in reaction to external changes.

\subsection{Definitions regarding Incremental Execution}


As a result of edit operations by a user, a model $M$ in the model slot of an \eGDN{} can undergo changes. In this context, a change corresponds to the creation or deletion of a vertex or an edge and is characterized by an atomic model delta of one of four types:

\begin{itemize}
\item$\delta^V_+$ is a single-element tuple $(v)$, with $v$ a vertex; applying $\delta^V_+$ to $M$ modifies $M$ into $M' = (V^M \cup \{v\}, E^M, s^M, t^M)$
\item$\delta^V_-$ is a single-element tuple $(v)$, with $v \in V^M$, applying $\delta^V_-$ to $M$ modifies $M$ into $M' = (V^M \setminus \{v\}, E^M, s^M, t^M)$
\item$\delta^E_+$ is a tuple $(e, s, t)$, with $e$ an edge and $s, t \in V^M$; applying $\delta^E_+$ to $M$ modifies $M$ into $M' = (V^M, E^M \cup \{e\}, s^M \cup \{(e, s)\}, t^M \cup \{(e, t)\})$
\item$\delta^E_-$ is a single-element tuple $(e)$, with $e \in E^M$; applying $\delta^E_-$ to $M$ modifies $M$ into $M' = (V^M, E^M \setminus \{e\}, s^M \setminus \{(e, s^M(e))\}, t^M \setminus \{(e, t^M(e))\})$
\end{itemize}

Importantly, this notion of atomic deltas can also cover the case of changes to attribute values in models in the form of typed attributed graphs \cite{heckel2002confluence}. In this context, attributes can be modeled via dedicated vertices representing attribute values and edges representing the assignment of these values to attributes of regular vertices.

Note that we do not allow implicit deletion of edges. If a vertex is deleted, it must not have any adjacent edges, that is, all adjacent edges have to be deleted previously. Similarly, if an edge is created, adjacent vertices have to be present in the model already.

Changes to the assignment set $A$ in a slot $s$ can similarly be described by atomic slot deltas:

\begin{itemize}
\item$\delta^A_+$ is a single-element tuple $(a)$, where $a \in dom(s)$; applying $\delta^A_+$ to the assignment set $A$ modifies $A$ into $A' = A \cup \{a\}$ 
\item$\delta^A_-$ is a single-element tuple $(a)$, where $a \in dom(s)$; applying $\delta^A_-$ to the assignment set $A$ modifies $A$ into $A' = A \setminus \{a\}$ 
\end{itemize}

For a slot node $s$, we denote the set of all possible atomic deltas over $dom(s)$ by $dom_\Delta(s)$ and the set of all possible sequences of elements in $dom_\Delta(s)$ by $\mathbb{S}(dom_\Delta(s))$.

Atomic deltas can be applied to a model or assignment set via an $apply$ procedure. We overload this procedure to also work with a sequence of atomic deltas, in which case the procedure applies the individual deltas in the order specified by the sequence.

We say that a sequence of atomic deltas $\Delta$ is \emph{minimal} for the contents of a slot node $s$, iff for all possible contents $v \in dom(s)$, it holds that $\nexists \Delta' \in \mathbb{S}(dom_\Delta(s)): apply(v, \Delta') = apply(v, \Delta)$, where we only consider equality of graphs up to isomorphism.


\OUT{\eGDN{} allows a naive reaction to model changes by simply performing a full recomputation of all operation node's results, that is, the contents of the \eGDN{}s slot nodes, via $populate$. This approach can however be unnecessarily expensive. Hence, an incremental reaction that only updates the contents of slot nodes instead is desirable.}

To enable reacting to model changes with an \eGDN{} $G = (O, S, E, s, t)$, an operation node $o \in O$ with input slots $in(o) = \{s_{i_1}, ..., s_{i_k}\}$ and output slots $out(o) = \{s_{o_1}, ..., s_{o_l}\}$ can be equipped with an $update$ procedure. This procedure is parametrized with a valuation function for $G$ and realizes a function
$\gamma^\delta : dom(s_{i_1}) \times ... \times dom(s_{i_k}) \times \mathbb{S}(dom_\Delta(s_{i_1})) \times ... \times \mathbb{S}(dom_\Delta(s_{i_k})) \times dom(s_{o_1}) \times ... \times dom(s_{o_l}) \rightarrow \mathbb{F}_\Delta$, with $\mathbb{F}_\Delta$ the set of functions $f_\Delta: out(o) \rightarrow \bigcup_{s_{o_i} \in out(o)} \mathbb{S}(dom_\Delta(s_{o_i}))$ such that $\forall s_{o_i} \in out(o) : f_\Delta(s_{o_i}) \in \mathbb{S}(dom_\Delta({s_{o_i}}))$.

To store deltas to react to later, $o$ is also extended by an array $o.\Delta$ that caches sequences of atomic deltas for its input and output slots, which can in practice be collected via a notification mechanism and the observer design pattern \cite{gamma1995design}. 

Calling $o.update(val)$ with $val$ a valuation function for $G$'s slots then yields the value of $\gamma^\delta$ parametrized according to $val$ and the cached sequences of deltas: $o.update(val) = \gamma^\delta(val(s_{i_1}), ..., val(s_{i_k}), o.\Delta[s_{i_1}], ..., o.\Delta[s_{i_k}], val(s_{o_1}), ..., val(s_{o_l}))$.

Intuitively, the $update$ procedure of an operation node should produce a sequence of deltas for the node's output slots that update the contents of these output slots to be consistent with the current contents of the operation node's input slots. Therefore, in addition to the contents of slots adajcent to $o$, an $update$ procedure may also consider additional information in the form of deltas to input slots to enable a more efficient realization.

Formally, an update procedure $update$ of operation node $o$ with input slots $in(o) = \{s_{i_1}, ..., s_{i_k}\}$, output slots $out(o) = \{s_{o_1}, ..., s_{o_l}\}$, and associated function $\gamma^\delta$ is correct iff for parameters $\Delta_1 \in \mathbb{S}(dom_\Delta(s_{i_1})), ..., \Delta_k \in \mathbb{S}(dom_\Delta(s_{i_k})$), $v_{i_1} \in dom(s_{i_1}), ..., v_{i_k} \in dom(s_{i_k})$, and $v_{o_1} \in dom(s_{o_1}), ..., v_l \in dom(s_{o_l})$,

\begin{equation*}
\begin{split}
&\exists f \in \gamma_S(v_{i_1}, ..., v_{i_k}): \forall s_{o_i} \in out(o): apply(v_{o_i}, f_\Delta(s_o)) = f(s_{o_i}),
\end{split}
\end{equation*}

with $f_\Delta = \gamma^\delta(\Delta_1, ..., \Delta_k, v_{i_1}, ..., v_{i_k}, v_{o_1}, ..., v_{o_l})$ and $\forall i \in [1, k], j \in [1, l] : s_{i_i} = s_{o_j} \rightarrow v_{i_i} = v_{o_j}$.

In many cases, the efficient realization of an $update$ procedure requires a relaxed notion of correctness, which requires the contents of the output slot to be consistent with the contents of the input slots before the application of the deltas according to $o$'s semantics function. In the following, we will refer to this relaxed notion of correctness as \emph{conditional correctness}, which is formally given by

\begin{flalign*}
&\exists v'_{i_1} \in dom(s_{i_1}), ..., v'_{i_k} \in dom(s_{i_k}):
\\
&\quad(apply(v'_{i_1}, \Delta_1) = v_{i_1} \wedge ... \wedge apply(v'_{i_k}, \Delta_k) = v_{i_k} \wedge
\\
&\quad\exists f' \in \gamma_S(v'_{i_1}, ..., v'_{i_k}):
\\
&\quad\quad(\forall s_{i_i} \in in(o) \cap out(o): v'_{i_i} = f'(s_{i_i}) \wedge
\\
&\quad\quad\forall s_{o_i} \in out(o) \setminus in(o): v_{o_i} = f'(s_{o_i})))
\\
&\rightarrow
\\
&\exists f \in \gamma_S(v_{i_1}, ..., v_{i_k}): \forall s_{o_i} \in out(o): apply(v_{o_i}, f_\Delta(s_o)) = f(s_{o_i}),
\end{flalign*}

with $f_\Delta = \gamma^\delta(\Delta_1, ..., \Delta_k, v_{i_1}, ..., v_{i_k}, v_{o_1}, ..., v_{o_l})$ and $\forall i \in [1, k], j \in [1, l] : s_{i_i} = s_{o_j} \rightarrow v_{i_i} = v_{o_j}$.

\OUT{For any operation node without a dedicated $update$ procedure, a generic na\"ive $update$ procedure can be constructed from the node's $valid$ and $populate$ procedures. Therefore, it is simply checked whether $valid$ returns $\texttt{true}$. If this is not the case, a full recomputation of the node's results is performed via $populate$ and deltas are constructed for each output slot via a comparison of its current contents to the result of the call to $populate$.}

We say that the realization of an $update$ procedure of an operation node $o$ is \emph{fully incremental} iff for a valuation function $val$ and cached deltas $\Delta_1, ..., \Delta_k$ with $\sum_{i \in \{1, ..., k\}} |\Delta_i| = 1$, that is, a single atomic delta as an input, (i) the runtime complexity is in $O(|\Delta_o|)$, with $\Delta_o = \bigcup_{s_{o} \in out(o)} o.update(val)(s_o)$, and (ii) the produced sets of deltas for each output slot are minimal.

\OUT{Analogously to $populate$ procedures, we can define non-recursiveness and potential update directions for $update$ procedures.}

In some cases, as with bidirectional or in-place model transformations, operation nodes may be connected to a slot via both an incoming and an outgoing edge, making such a slot simultaneously an input and output slot to the same operation node. Such an operation node may as a result exhibit recursive behavior, since an application of its $update$ procedure can also change the contents of the operation node's input slots and thus necessitate further calls to $update$ to restore consistency. In this context, we call an $update$ procedure of an operation node $o$ is \emph{non-recursive}, if, after one execution of $o$'s update function and subsequent application of the resulting deltas to $o$'s output slot values, a second execution with updated slot values never yields any new deltas.

Formally, an $update$ procedure of an operation node $o$ with input slots $in(o) = \{s_{i_1}, ..., s_{i_k}\}$ and output slots $out(o) = \{s_{o_1}, ..., s_{o_l}\}$, is \emph{non-recursive}, if for any possible parametrization $v_{i_1} \in dom(s_{i_1}), ..., v_{i_k} \in dom(s_{i_k})$, $\Delta_1 \in \mathbb{S}(dom_\Delta(s_{i_1})), ..., \Delta_k \in \mathbb{S}(dom_\Delta(s_{i_k}))$, and $v_{o_1} \in dom(s_{o_1}), ..., v_ {o_l} \in dom(s_{o_l})$, it holds that

\begin{equation*}
\forall s_o \in out(o) : \gamma^\delta(\Delta_1', ..., \Delta_k', v_{i_1}', ..., v_{i_k}', v_{o_1}', ..., v_{o_l}')(s_o) = \emptyset,
\end{equation*}

where

\begin{equation*}
	\Delta_i' = 
	\begin{cases}
		f_\Delta(s_i) &\quad\text{if } s_{i_i} \in out(o)\\
		\Delta_i &\quad\text{otherwise}
	\end{cases}
\end{equation*}

and

\begin{equation*}
	v_{i_i}' = 
	\begin{cases}
		apply(v_{i_i}, f_\Delta(s_i)) &\quad\text{if } s_{i_i} \in out(o)\\
		v_{i_i} &\quad\text{otherwise}
	\end{cases}
\end{equation*}

and

\begin{equation*}
	v_{o_i}' = apply(v_{o_i}, f_\Delta(s_{o_1})),
\end{equation*}

with $f_\Delta = \gamma^\delta(\Delta_1, ..., \Delta_k, v_{i_1}, ..., v_{i_k}, v_{o_1}, ..., v_{o_l})$.

The \emph{potential update directions} of an update procedure of operation node $o$ for a set of input slots $S_i \subseteq in(o)$ are given by $o.dir_\Delta(o, S_i)$, where for a slot $s_o \in out(o)$,

\begin{equation*}
\begin{split}
s_o \in o.dir_\Delta(o, S_i) \leftrightarrow & \exists \Delta_1 \in \mathbb{S}(dom_\Delta(s_{i_1})), ..., \Delta_k \in \mathbb{S}(dom_\Delta(s_{i_k})), \\
	&\quad v_{i_1} \in dom(s_{i_1}), ..., v_{i_k} \in dom(s_{i_k}), \\
	&\quad v_{o_1} \in dom(s_{o_1}), ..., v_{o_l} \in dom(s_{o_l}) : \\
		&\quad\quad \forall s_{i_i} \in in(o) \setminus S_i : \Delta_i = \emptyset \wedge \\
		&\quad\quad \gamma^\delta(\Delta_1, ..., \Delta_k, v_{i_1}, ..., v_{i_k}, v_{o_1}, ..., v_{o_l})(s_o) \neq \emptyset
\end{split}
\end{equation*}

Intuitively, $o.dir_\Delta(o, S_i)$ thus denotes the subset of output slots for which $o$'s update procedure may generate deltas if the contents of at most the input slots in $S_i$ have changed.

A function $dir_\Delta$ for potential update directions is monotonic by definition in the sense that $\forall S_{i_1}, S_{i_2} \subseteq in(o): S_{i_1} \subseteq S_{i_2} \rightarrow o.dir_\Delta(o, S_{i_1}) \subseteq o.dir_\Delta(o, S_{i_2})$. We say that $dir_\Delta$ is \emph{union monotonic} if it furthermore holds that $\forall S_{i_1}, S_{i_2} \subseteq in(o): o.dir_\Delta(S_{i_1}) \cup o.dir_\Delta(S_{i_2}) = o.dir_\Delta(S_{i_1} \cup S_{i_2})$.

In the following, we present algorithms for the incremental execution of an \eGDN{} based on the $update$ procedures of its operation nodes. For these algorithms, we assume that deltas cached in the input \eGDN{} are consistent in the sense that they correspond to a modification from slot contents that were consistent with the semantics functions of all operations in the \eGDN{} to the current contents. Intuitively, this assumption simply implies that the presented algorithms can only produce consistent slot contents if the slot contents were previously consistent at some point and all changes since then have been tracked and cached in the \eGDN{}.


\subsection{Incremental Execution with Guaranteed Termination}
 
Given a correct $update$ function for each operation node, an input \eGDN{} $G = (O, S, E, s, t)$ can be executed incrementally in the context of a valuation function $val$ via Algorithm \ref{algo:execute_incremental_order}. Therefore, Algorithm \ref{algo:execute_incremental_order} first derives an ordering of $G$'s operation nodes and then updates the $val$ function by executing the nodes' $update$ functions, applying the resulting deltas to the appropriate slots, and updating the cached deltas.

Importantly, the employed ordering has to guarantee correct results in the sense that the contents of $G$'s slots after the execution must be consistent with the semantics functions of all of its operation nodes, that is, it must hold that $\forall o \in O : o.valid(val)$.

If $G$ takes the form of a directed acyclic graph and operation nodes do not share output slots, such an ordering can be obtained by simply sorting $G$'s operation nodes topologically. However, requiring DAG structure represents a substantial restriction, as it effectively prohibits bidirectional transformations where some input slots are also output slots. Moreover, the assumption regarding the complete absence of shared output slots, while required to prevent overwriting of operation's results, is another obstacle to realizing several desirable use cases, for instance those involving chains of bidirectional transformations.

Based on the properties of an \eGDN{}'s operation nodes with respect to non-recursiveness and potential update directions, an appropriate order can also be found for certain cyclical \eGDN{}s, with a relaxed assumption regarding shared output slots. Algorithm \ref{algo:find_incremental_order} represents an analysis for an \eGDN{} $G$ that contains only nodes with non-recursive $update$ procedures and a set of slots $S_i$ with initially modified contents. If successful, the algorithm returns an execution order that can be used instead of the topological ordering in Algorithm \ref{algo:execute_incremental_order}. Importantly, the computed ordering still yields a valuation function that is consistent with all operations' semantics.

\SetKwFunction{ExecuteIncrementalDAG}{ExecuteIncrementalDAG}
\SetKwFunction{FindValidUpdateOrder}{FindValidUpdateOrder}
\SetKwFunction{Apply}{apply}
\begin{algorithm}
\LinesNumbered
\myproc{\ExecuteIncrementalDAG{$G = (O, S, E, s, t), val$}} {
	\Input{$G$: The \eGDN\\
	$val$: A valuation function for $G$'s slots}
	\BlankLine
	
	$D \leftarrow \FindValidUpdateOrder{$O, \{s \in S | \exists o \in O : o.\Delta[s] \neq \emptyset\}$}$\;
	\If{$D \neq \textbf{null}$} {
		\ForEach{$o \in D$} {													\label{line:execute_incremental_order_main_loop}
			$\Delta_o \leftarrow o.update(val)$\;
			\ForEach{$s_o \in out(o)$} {
				$val(s_o) \leftarrow \Apply(val(s_o), \Delta_o(s_o))$\;
				\ForEach{$o' \in out(s_o)$} {
					$o'.\Delta[s_o] \cup \Delta_o$\;
				}
			}
			\ForEach{$s \in in(s) \cup out(s)$} {
				$o.\Delta[s] \leftarrow \emptyset$\;
			}
		}
	}
}
\BlankLine
\caption{Incremental algorithm for executing an \eGDN{} based on an ordering of its operation nodes} \label{algo:execute_incremental_order}
\end{algorithm}

\SetKwFunction{SortTopologically}{SortTopologically}
\begin{algorithm}
\LinesNumbered
\myproc{\FindValidUpdateOrder{$G = (O, S, E, s, t), S_i$}} {
	\Input{$G$: The \eGDN\\
		$S_i$: The set of initially changed slots}
	
	$C \leftarrow \textbf{new }Array(|O|)$\;
	$C.init(\emptyset)$\;
	$Q \leftarrow \textbf{new }Queue$\;
	$G_T = \textbf{new }Graph$\;
	
	\BlankLine
	
	\ForEach{$o \in in(S_i) \cup out(S_i)$} {											\label{line:find_incremental_order_init_loop}
		$Q.enqueue(o)$\;
		$C[o] \leftarrow S_i \cap in(o)$\;
	}
	$G_T.addVertices(Q)$\;
	\While{$\neg Q.isEmpty()$} {													\label{line:find_incremental_order_main_loop}
		$o \leftarrow Q.dequeue()$\;												\label{line:find_incremental_order_dequeue}
		$S_o \leftarrow o.dir_\Delta(o, C[o])$\;
		$O_o \leftarrow out(S_o) \cup in(S_o) \setminus \{o\}$\;
		\ForEach{$o' \in O_o$} {													\label{line:find_incremental_order_queue_loop}
			\If{$\neg o' \in Q$} {
				$Q.enqueue(o')$\;
			}
			$C[o'] \leftarrow C[o'] \cup (S_o \cap in(o'))$\;
			$G_T.addVertexIfNotExists(o')$\;
			$G_T.createEdgeIfNotExists(o, o')$\;
			\If{$G_T.hasCycle()$} {
				\Return \textbf{null}\;
			}
		}
		$C[o] \leftarrow \emptyset$\;
	}
	
	\Return \SortTopologically{$G_T$}\;
}
\BlankLine
\caption{Static analysis algorithm for finding an \eGDN{} update order} \label{algo:find_incremental_order}
\end{algorithm}

The algorithm first creates an array $C$ with one cell per operation node in $O$ and initializes it with empty sets. It also initializes a queue $Q$ with all operation nodes that are connected to a slot in $S_i$ and, for each such operation node, stores the set of its input slots that are also in $S_i$ in the corresponding cell in $C$. Then, a slightly modified breadth-first search is performed over the \eGDN{} structure using the initialized queue $Q$ to essentially simulate an execution of $G$ without concrete inputs.

Therefore, the procedure loops until $Q$ is empty. In each loop execution, the first operation node $o$ in $Q$ is dequeued. Then, all output slot nodes for which deltas could be produced due to the execution of $o'$s update procedure $S_o$ are obtained based on $o$'s potential update directions and the set of slots that might currently contain unhandled deltas, which is retrieved from $C$. Afterwards, all operation nodes $o'$ connected to a slot in $S_o$ are added to $Q$ if they are not yet contained. Also, the set of $o$'s input slots with potentially unhandled deltas stored in $C$ is updated based on $S_o$. An exception is made for the currently considered node $o$, which is never added to the queue again and whose set of input slots with potentially unhandled deltas is reset to the empty set, exploiting the assumption that all $update$ procedures in the \eGDN{} are non-recursive.

During execution, the algorithm keeps track of the dependencies between $G$'s operations in a trigger graph $G_T$. Execution aborts by returning \textbf{null} as soon as a cyclical dependency is detected, which may indicate a potential infinite loop in $G$'s execution for the initially populated slots $S_i$. This also guarantees that after a full execution of the loop in line \ref{line:find_incremental_order_main_loop}, $G_T$ is a DAG.

Finally, a topological ordering of $G_T$, is returned as a possible canonic execution order that, under the mentioned assumptions, produces a valuation function for the input \eGDN{}'s slots that is consistent with the semantics functions of all of the \eGDN{}'s operation nodes.

While the presented algorithm is formulated to handle incremental changes to a network of models and model operations, the batch case that requires an initial execution of model operations to derive corresponding query results and transformed models for an initial set of existing models can be handled in a straightforward manner. Therefore, an incremental construction of the initially existing models can be emulated by deriving trivial sequences of corresponding creation operations, which can act as the starting point for the algorithm. This only requires the assumption that the case where all slots of an \eGDN{} are empty constitutes a consistent valuation regarding the semantics of all of the \eGDN{}'s operations, which seems reasonable. The additional assumption is essentially required to satisfy the rerquirement regarding consistency of initially cached deltas with the current state.

\subsubsection*{Termination}

By including the additional termination criterion in the loop in line \ref{line:find_incremental_order_main_loop} of Algorithm \ref{algo:find_incremental_order} that requires the constructed dependency graph to be acyclic, Algorithm \ref{algo:find_incremental_order} is guaranteed to terminate.

\begin{theorem} \label{the:find_incremental_order_termination}
Algorithm \ref{algo:find_incremental_order} always terminates.
\end{theorem}

\begin{proof}
Except for the loop in line \ref{line:find_incremental_order_main_loop}, all loops only iterate over finite sets, and all individual operations always terminate. The loop in line \ref{line:find_incremental_order_main_loop} also always terminates due to the termination criterion regarding cyclical dependencies between the \eGDN{}'s operation nodes: Since one operation node is removed from $Q$ in each loop iteration, termination is only threatened if operation nodes keep getting added to $Q$. Since there is only a finite number of operation nodes, infinite behavior can only occur as a result of cycles in the modified breadth-first search. However, such cycles are detected via $G_T$ and immediately lead to abortion of the execution.
\end{proof}

Consequently, Algorithm \ref{algo:execute_incremental_order} is also guaranteed to terminate if the execution of the input \eGDN{}'s $update$ procedures always terminates.

\begin{theorem}
For an input \eGDN{} $G$, Algorithm \ref{algo:execute_incremental_order} always terminates if the $update$ procedures of $G$'s operation nodes always terminate.
\end{theorem}

\begin{proof}
According to Theorem \ref{the:find_incremental_order_termination}, Algorithm \ref{algo:find_incremental_order} always terminates by either aborting or returning a sequence of operation nodes. Such a sequence being returned implies that the sequence is finite. The loop in line \ref{line:execute_incremental_order_main_loop} is thus only executed for finitely many iterations. Since all other loops only iterate over finite sets and all individual operations always terminate due to the assumption regarding $G$'s $update$ procedures, Algorithm \ref{algo:execute_incremental_order} always terminates.
\end{proof}

\subsubsection*{Correctness}

The following theorem states the correctness of a canonic execution order resulting from an execution of Algorithm \ref{algo:find_incremental_order} for the case that all $dir_\Delta$ functions are union monotonic.

\begin{theorem} \label{the:correctness_incremental_order}
For inputs $G = (O, S, E, s, t)$ and $val$, if all $update$ procedures in $G$ are correct and non-recursive, all $dir_\Delta$ functions in $G$ are union-monotonic, and if the valuation function before the application of the deltas cached in $G$ was consistent with the semantics of $G$'s operation nodes, Algorithm \ref{algo:execute_incremental_order} aborts or produces a final valuation function $val$ such that $\forall o \in O : o.valid(val)$.
\end{theorem}

\begin{proof}
If Algorithm \ref{algo:execute_incremental_order} does not abort, a canonic execution order $R$ for $G$'s operation nodes has been generated by topologically sorting the resulting directed acyclic dependency graph $G_T$ of a terminating execution of Algorithm \ref{algo:find_incremental_order}.

Due to the non-recursiveness of $G$'s operation nodes, we know that after executing an operation node $o$ via Algorithm \ref{algo:execute_incremental_order}, it holds that $o.valid(val)$. Thus, for an operation node $o$, $\neg o.valid(val)$ can only hold after executing the entire sequence $R$ if there exists some operation node $o$' that comes after $o$ in $R$ and that changes the contents of a slot adjacent to $o$ or if $o \notin R$. Considering that all operation nodes that have an adjacent slot with initially modified contents are initially added to $G_T$, the algorithm has terminated, and that prior to the cached modifications of $G$'s slots, slot contents were consistent with all operation nodes' semantics, for a node $o \notin R$, $\neg o.valid(val)$ can also only hold if there is some operation node $o' \in R$ that changes the contents of a slot $s$ adjacent to $o$.

In either case, we know that there cannot exist an edge from $o'$ to $o$ in $G_T$, because $o'$ either comes after $o$ in the topological ordering or because the addition of such an edge would have caused $o$ to be added to $G_T$ and consequently $R$. This means that, due to the definition of $o'.dir_\Delta$ and because of the assumed union monotonicity, there must be a slot $s'$ with $o'.\Delta[s] \neq \emptyset$ and $s \in o'.dir_\Delta(\{s'\})$ before executing $o'$ that was never in the set of slots $C[o']$ when $o'$ was dequeued in line \ref{line:find_incremental_order_dequeue} of Algorithm \ref{algo:find_incremental_order}. Since slots are only removed from $C[o']$ when $o'$ is dequeued and corresponding edges are added, we know that $s' \notin S_i$ and thus $o'.\Delta[s'] = \emptyset$ at the start of Algorithm \ref{algo:execute_incremental_order}, as otherwise, $o'$ would have been added to $Q$ and the edge between $o'$ and $o$ would eventually have been created.

There hence must be a node $o''$ that comes before $o'$ in $R$ that modified the contents of $s'$. Also for $o''$, there must be a slot $s''$ with $o''.\Delta[s''] \neq \emptyset$ and $s' \in o''.dir_\Delta(\{s''\})$ before executing $o''$ that was never in $C[o'']$ whenever $o''$ was dequeued (because otherwise, the edge between $o'$ and $o$ would have been created eventually). Therefore, again, there must be an operation node before $o''$ in $R$ that modified the contents of $s''$ and for which the same constraints apply as for $o''$. Ultimately, this implies that for the first operation node in the sequence, there must be a predecessor that changes the contents of some slot node, which is obviously a contradiction.

Hence, there cannot be an operation node in $R$ whose execution changes the contents of a slot adjacent to a previous operation node in $R$ or an operation node not contained in $R$. Consequently, we know that after executing $R$, $\forall o \in O : o.valid(val)$.
\end{proof}

If the \eGDN{}'s $update$ functions are only conditionally correct, an additional constraint has to be introduced regarding \eGDN{} structure to guarantee correctness. Namely, operation nodes may not share output slots if the output slots are not also input slots of all sharing operation nodes, and output slots of a node that are not simultaneously input slots of the same node may not have their contents modified by users.

Intuitively, these conditions impose the restriction on the \eGDN{} structure that the contents of an operation node's output slot may not be modified by another operation node or a user without that operation node being able to pick up on and handle the changes.

\begin{corollary}
Assuming that $\forall o_1, o_2 \in O : o_1 \neq o_2 \rightarrow \forall s \in out(o_1) \cap out(o_2) : s \in in(o_1) \wedge s \in in(o_2)$ and $\forall o \in O : \forall s \in out(o) : o.\Delta[s] = \emptyset$, for inputs $G = (O, S, E, s, t)$ and $val$, if all $update$ procedures in $G$ are conditionally correct and non-recursive, all $dir_\Delta$ functions in $G$ are union-monotonic, and if the valuation function before the deltas cached in $G$ was consistent with the semantics of $G$'s operation nodes, Algorithm \ref{algo:execute_incremental_order} aborts or produces a final valuation function $val$ such that $\forall o \in O : o.valid(val)$.
\end{corollary}

\begin{proof}
From the additional assumptions regarding output slots of $G$'s operation nodes it follows directly that the condition in the definition of conditional correctness is never violated. Thus, the statement from Theorem \ref{the:correctness_incremental_order} also applies for the case of conditionally correct $update$ functions.
\end{proof}

Notably, the order in which operation nodes are added to the queue $Q$ in lines \ref{line:find_incremental_order_init_loop} and \ref{line:find_incremental_order_queue_loop} of Algorithm \ref{algo:find_incremental_order} is undefined. Since the order of operation nodes in $Q$ affects the behavior of the algorithm, this might mean that Algorithm \ref{algo:find_incremental_order} is ultimately not deterministic.

We can however show that, if Algorithm \ref{algo:find_incremental_order} does not abort due to cycles in $G_T$, the final dependency graph $G_T$ is uniquely defined, independently of the order in which operation nodes are added to $Q$. Thus, the only remaining nondeterminism in Algorithm \ref{algo:find_incremental_order} affecting the result stems from the topological sorting at the end of the algorithm, which is an inherently nondeterministic operation.

\begin{theorem} \label{the:find_incremental_order_determinism}
For inputs $G = (O, S, E, s, t)$ and $S_i$, the dependency graph $G_T$ after a full execution of the loop in line \ref{line:find_incremental_order_main_loop} of Algorithm \ref{algo:find_incremental_order} is uniquely defined up to isomorphism if all $dir_\Delta$ functions in $G$ are union monotonic.
\end{theorem}

\begin{proof}
The set of vertices initially added to $G_T$ is uniquely determined by $in(S_i) \cup out(S_i)$. Since additional vertices are only ever added in conjunction with the creation of an edge, the set of vertices added during the execution of the loop in line \ref{line:find_incremental_order_main_loop} is determined by the set of added edges.

To show the unique determination of added edges by the algorithm's inputs, we show that in a terminating execution of the loop, the initial set $S_i$ in conjunction with the \eGDN{} $G$ uniquely determines a set of pairs of operation nodes $(o_1, o_2)$, between which directed edges are created in $G_T$.

$S_i$ uniquely determines the set of operation nodes $O_Q = in(S_i) \cup out(S_i)$ that is initially added to $Q$. For each of these operation nodes $o_Q \in O_Q$, due to the monotonicity of $o_Q.dir_\Delta$ and because slots are only removed from $C[o_Q]$ after $o_Q$ has been dequeued and processed, at least the edges for pairs $edges_S(o_Q, S_i) = \{(o_Q, o_T) | o_T \in out(S_o) \cup in(S_o) \setminus \{o_Q\}\}$ are added to $G_T$, where $S_o = o_Q.dir_\Delta(S_i \cap in(o_Q))$ when $o_Q$ is dequeued. According to the assumption regarding union monotonicity, we can also write $edges_S(o_Q, S_i) = edges_\emptyset(o_Q) \cup \bigcup_{s_i \in S_i \cap in(o_Q)} edges_N(o_Q, s_i)$, with $edges_\emptyset(o_Q) = \{(o_Q, o_T) | o_T \in out(o_Q.dir_\Delta(\emptyset)) \cup in(o_Q.dir_\Delta(\emptyset)) \setminus \{o_Q\}\}$ and $edges_N(o_Q, s_i) = \{(o_Q, o_T) | o_T \in out(o_Q.dir_\Delta(in(o_Q) \cap \{s_i\})) \cup in(o_Q.dir_\Delta(in(o_Q) \cap \{s_i\})) \setminus \{o_Q\}\}$.

In addition, the modification of $C$ and $Q$ that takes place for each dequeued $o_Q \in O_Q$ may cause the addition of further edges down the line. Specifically, for each $s_o \in o_Q.dir_\Delta(\{s_i\})$ and each $o_T \in out(s_o) \cup in(s_o) \setminus \{o_Q\}$, $o_T$, if not already contained, is added to $Q$ and subsequently handled in the same way as $o_Q$, with $s_i$ guaranteed to be in $C[o_T]$ at that moment. This will cause the addition of all edges corresponding to the pairs $edges_N(o_T, s_o)$ and again trigger the addition of further edges. Due to the monotonicity of $o_T.dir_\Delta$ and because slots are only removed from $C[o_T]$ when $o_T$ is dequeued, the addition of these edges happens independently from any other modifications to $C[o_T]$ that might be made in the meantime. Furthermore, due to the assumption regarding union monotonicity of $o_T.dir_\Delta$, a combination of modifications of $C[o_T]$ cannot yield any additional edges compared to what is yielded for the individual members of $C[o_T]$.

Because neither can $C[o_T]$ be modified in any other way, nor can edges be added to $G_T$ in any other way, the set of pairs of operation nodes $(o_1, o_2)$ between which directed edges are created in $G_T$ is given by the function $edges(S_i) = \bigcup_{o_Q \in in(S_i) \cup out(S_i)}(edges_\emptyset(o_Q) \cup \bigcup_{s_i \in S_i \cap in(o_Q)} edges_R(o_Q, s_i))$, with $edges_R(o_Q, s_i) = edges_N(o_Q, s_i) \cup \bigcup_{o_T \in O_T)) \setminus \{o_Q\}}\bigcup_{s_o \in o_Q.dir_\Delta(in(o_Q) \cap \{s_i\}} edges_R(o_T, s_o)$, where $O_T$ is given by $O_T = out(o_Q.dir_\Delta(in(o_Q) \cap \{s_i\})) \cup in(o_Q.dir_\Delta(in(o_Q) \cap \{s_i\}$.

The loop terminating due to $Q$ becoming empty implies that all nodes ever added to $Q$ have been processed and hence all corresponding edges have been added to $G_T$. Since it is ensured that for each pair of operation nodes $(o_1, o_2)$, only one corresponding edge is added, we know that regardless of the concrete processing order, $G_T$ always contains exactly one directed edge for each pair $(o_1, o_2) \in edges(S_i)$.

Since the set of added vertices is uniquely determined by the set of added edges and each vertex can only be added once, the set of $G_T$'s vertices is uniquely defined for inputs $G$ and $S_i$.

The graph $G_T$ at the end of a full execution of the loop in line \ref{line:find_incremental_order_main_loop} of Algorithm \ref{algo:find_incremental_order} is hence uniquely defined for inputs $G$ and $S_i$, regardless of the order in which operation nodes are added to $Q$ in lines \ref{line:find_incremental_order_init_loop} and \ref{line:find_incremental_order_queue_loop}.
\end{proof}

The fact that $G_T$ is uniquely defined by the inputs $G$ and $S_i$ also implies that if an execution of Algorithm \ref{algo:execute_incremental_order} terminates without aborting, so does any possible execution for the same inputs.

\begin{theorem}
An execution of the loop in line \ref{line:find_incremental_order_main_loop} of Algorithm \ref{algo:execute_incremental_order} terminates without aborting for inputs $G = (O, S, E, s, t)$ and $S_i$ if and only if any other execution for the same inputs also terminates without aborting.
\end{theorem}

\begin{proof}
According to Theorem \ref{the:find_incremental_order_termination}, the loop in line \ref{line:find_incremental_order_main_loop} of Algorithm \ref{algo:execute_incremental_order} always terminates, either because of a violation of the looping condition or because the loop aborts. Since the loop aborts if and only if a cycle is detected in $G_T$ at any point and edges are never removed from $G_T$, it follows that the loop terminates without aborting if and only if the set of edges added to $G_T$ during the loop execution does not form cycles. Since the set of edges added to $G_T$ during the loop execution is functionally determined by only the inputs $G$ and $S_i$, it hence follows that, if an execution of the loop terminates without aborting for $G$ and $S_i$, any execution with the same inputs will also terminate without aborting.
\end{proof}

Furthermore, we can show that if there exists an execution sequence for $G$ that guarantees correct results in the worst case and that executes every operation node at most once, Algorithm \ref{algo:find_incremental_order} finds such a sequence.

\begin{theorem} \label{the:find_incremental_order_completeness}
For an input \eGDN{} $G = (O, S, E, s, t)$ with correct and non-recursive $update$ procedures with union monotonic $dir_\Delta$ functions and a set of slots $S_i \subseteq S$ with initially modified contents for a valuation function $val$, assuming that

\begin{enumerate}
\item for any operation node $o_1 \in O$, for any execution of $o_1.update(val')$ with a valuation function $val'$ and deltas for input slots $S_\Delta$, it holds that $\forall s_o \in o_1.dir_\Delta(S_\Delta) : o_1.update(val')(s_o) \neq \emptyset$,
\item for a second node $o_2 \in O$ with $o_1 \neq o_2$, it holds that $\exists s_o \in out(o_1) \cap in(o_2) : o_1.update(val')(s_o) \neq \emptyset \rightarrow \neg o_2.valid(val'')$, where for $s \in S$
\begin{equation}
val''(s) =
	\begin{cases}
	apply(val'(s), o_1.update(val')(s)) &\quad \text{if } s \in out(o_1) \cap in(o_2)\\
	val'(s) &\quad \text{otherwise}
	\end{cases}
\end{equation}
and
\item it holds that $\forall o \in in(S_i) \cup out(S_i) : \neg o.valid(val)$,
\end{enumerate}

\noindent
if there exists a sequence that guarantees a correct resulting valuation function if executed via Algorithm \ref{algo:execute_incremental_order} and that only contains each node $o \in O$ once, Algorithm \ref{algo:find_incremental_order} returns such a sequence.
\end{theorem}

\begin{proof}
Under the given assumptions, the set of edges in $G_T$ created by Algorithm before termination or abortion represents a subset of all relations between pairs of operation nodes $(o_1, o_2)$, where $o_1$'s $update$ procedure has to be executed at least once to produce a correct final valuation function and that execution modifies the contents of a slot adjacent to $o_2$, necessitating the subsequent execution of $o_2$ according to assumption (2).

This is due to the fact that, to restore consistency, all operation nodes in $o \in in(S_i) \cup out(S_i)$ have to be executed at least once according to assumptions (2) and (3). All these operation nodes $o_1$ are initially added to the queue $Q$ in Algorithm \ref{algo:find_incremental_order}. Each execution of an operation node $o_1$, according to assumption (1), modifies all slots in $o_1.dir_\Delta(S_i \cap in(o_1))$, which necessitates a subsequent execution of all operation nodes $o_2 \in in(o_1.dir_\Delta(S_i \cap in(o_1))) \cup out(o_1.dir_\Delta(S_i \cap in(o_1)))$ according to assumption (2). Algorithm \ref{algo:find_incremental_order} creates edges for all these pairs $(o_1, o_2)$ when $o_1$ is dequeued.

The subsequent execution of any operation node $o_2$ similarly necessitates the execution of all nodes $o_3 \in in(o_2.dir_\Delta(S_i \cap in(o_2))) \cup out(o_2.dir_\Delta(S_i \cap in(o_2)))$, which is also reflected by the edges created in Algorithm \ref{algo:find_incremental_order} when $o_2$ is dequeued, and so on. Since the algorithm creates no additional edges due to the assumption regarding union monotonicity of the $dir_\Delta$ functions, all edges in $G_T$ represent such necessary relationships on the ordering of operation nodes\footnote{As a side note, since operation nodes can be dequeued/executed with different sets of potentially modified input slots, an edge between nodes $(o_1, o_2)$ in $G_T$ does not necessarily mean that $o_2$ has to be executed after \emph{any} execution of $o_1$, but only that such a subsequent execution is necessary \emph{at least once}.}.

Since Algorithm \ref{algo:find_incremental_order} always produces a correct sequence of operation nodes if $G_T$ is acyclic, we can assume that in the case where the algorithm does not produce an ordering, there is at least one cycle in $G_T$. There hence cannot exist a sequence of the operation nodes involved in this cycle where each node is only contained once and each node is executed at least once after its predecessor in the cycle. Thus, by contraposition it follows that, if there exists a sequence of operation nodes that guarantees correct results and where each operation node is only contained once, Algorithm \ref{algo:find_incremental_order} finds such a sequence.
\end{proof}

Note that there may be finite orders of operation node executions that guarantee correct results based on the assumptions in Theorem \ref{the:find_incremental_order_completeness} that are not found by Algorithm \ref{algo:find_incremental_order}. However, these orders require that at least one operation node is executed at least twice.


\subsection{Incremental Execution of Arbitrary \eGDN{}s}

If the \eGDN{} is not a DAG and no suitable ordering of its operation nodes can be found via Algorithm \ref{algo:find_incremental_order}, incremental execution can instead be achieved via a simple fixpoint iteration as in Algorithm \ref{algo:execute_incremental}.

\SetKwFunction{ExecuteIncremental}{ExecuteIncremental}
\begin{algorithm}
\LinesNumbered
\myproc{\ExecuteIncremental{$G = (O, S, E, s, t), val$}} {
	\Input{$G$: The \eGDN\\
	$val$: A valuation function for G's slots}
	\BlankLine
	
	$D \leftarrow \{o \in O|\exists s \in in(o) \cup out(o) : o.\Delta[s] \neq \emptyset\}$\;
	\While{$D \neq \emptyset$} {										\label{line:execute_incremental_main_loop}
		$D_n \leftarrow \emptyset$\;
		\ForEach{$o \in D$} {
			$D \leftarrow D \setminus \{o\}$\;
			$\Delta_o \leftarrow o.update(val)$\;
			\ForEach{$s \in in(s) \cup out(s)$} {
				$o.\Delta[s] \leftarrow \emptyset$\;
			}
			\ForEach{$s_o \in out(o)$} {								\label{line:execute_incremental_operations_loop}
				\If{$\Delta_o(s_o) \neq \emptyset$} {
					$val(s_o) \leftarrow \Apply(val(s_o), \Delta_o(s_o))$\;
	
					\ForEach{$o' \in out(s_o)$} {						\label{line:execute_incremental_out_loop}
						$o'.\Delta[s_o] \cup \Delta_o$\;
						\If{$o' \notin D$}{
							$D_n \leftarrow D_n \cup \{o'\}$\;
						}
					}
					\ForEach{$o' \in in(s_o)$} {						\label{line:execute_incremental_in_loop}
						\If{$o' \neq o \wedge o' \notin D$}{
							$D_n \leftarrow D_n \cup \{o'\}$\;
						}
					}
				}
			}
		}
		$D \leftarrow D_n$\;
	}
}
\BlankLine
\caption{Incremental algorithm for \eGDN{} execution} \label{algo:execute_incremental}
\end{algorithm}

Algorithm \ref{algo:execute_incremental} first initializes the set of operation nodes that require execution $D$ with the set of all operation nodes in the input \eGDN{} for which there are changes in one of the node's input or output slots. Then, the algorithm iterates until a fixpoint is reached.

Therefore, a set of operation nodes that will require execution in the next iteration $D_n$ is initialized with the empty set. Afterwards, for each operation node $o$ that is due for execution in the current iteration, that node is removed from the set $D$. Then, $o$'s $update$ procedure is called to compute a set of changes to the contents of $o$'s output slots to make them consistent with the semantics of $o$.

For each output slot $s_o$ of $o$ that $update$ has computed changes for, these changes are subsequently applied and appropriately registered at each operation node $o'$ for which $s_o$ is an input slot. If any such $o'$ is not still due for execution in the current iteration, it is marked for execution in the next iteration by adding it to $D_n$. Operation nodes for which $s_o$ is an output slot are similarly marked for execution. Finally, after all operation nodes in $D$ have been considered, $D_n$ replaces $D$ and a new iteration starts if $D_n$ is not empty.

Analogously to Algorithm \ref{algo:execute_incremental_order}, Algorithm \ref{algo:execute_incremental} can handle the batch case of an initial \eGDN{} execution for existing models by encoding such existing models as sequences of element creations.

\subsubsection*{Termination}

In contrast to Algorithm \ref{algo:execute_incremental_order}, Algorithm \ref{algo:execute_incremental} is not guaranteed to terminate, since cyclical transitive dependencies of operation nodes may cause infinite cycles of changes to the contents of some slot node. Without restricting developers in what kinds of \eGDN{}s they are allowed to specify, this problem is inevitable.

In practice however, termination of networks of model operations like \eGDN{}s can be achieved despite the presence of cyclical structures. In some cases for instance, cycles at the network level do not necessarily correspond to actual cyclical dependencies of model operations if the involved model operations only affect distinct parts of slot contents, such as elements of certain, distinct types. In some cases, a restructuring of the \eGDN{} may remove cycles at the structural level while preserving semantics, for instance by converting in-place model transformations without an effective reflexive dependency into a model transformation with distinct input and output models.

Moreover, cycles of model operations may exhibit monotonic behavior, for instance by deleting certain elements in each iteration that are never recreated, thus guaranteeing convergence. Ultimately however, it remains the responsibility of the developers to create networks of model operations that do not lead to infinite loops in execution.

\subsubsection*{Correctness}

If Algorithm \ref{algo:execute_incremental} terminates, the resulting valuation function is guaranteed to be consistent with the semantics of all operation nodes in the input \eGDN{}.

\begin{theorem} \label{the:correctness_incremental}
For inputs $G = (O, S, E, s, t)$ and $val$, if Algorithm \ref{algo:execute_incremental} terminates, all employed $update$ procedures are correct and non-recursive, and if the valuation function before the application of the deltas cached in $G$ was consistent with the semantics of $G$'s operation nodes, the algorithm produces a final valuation function $val$ such that $\forall o \in O : o.valid(val)$.
\end{theorem}

\begin{proof}
We show that the invariant (1) $\forall o \in O : \neg o.valid(val) \rightarrow o \in D$ holds for the loop in line \ref{line:execute_incremental_main_loop} via induction over the number of loop iterations.

The base case for invariant (1) holds due to the initialization of $D$ and the assumption regarding the initial cached deltas and previous valuation function.

To show the induction step for invariant (1), we first show that under the induction assumption, the invariant (2) $\forall o \in O : \neg o.valid(val) \rightarrow o \in D \cup D_n$ holds for the loop in line \ref{line:execute_incremental_operations_loop}. This can also be done via induction.

The base case for invariant (2) holds due to the induction assumption of (1).

The induction step holds for invariant (2) since in each iteration of the inner loop, only one operation node $o$ is executed via its $update$ procedure and removed from $D$, updating $val$ and the cached deltas in the process. If the execution of $o$ does not change the contents of one of its own input slots, we know that afterwards, $o.valid(val)$ due to the assumption regarding correctness of $update$ procedures and because the cached deltas are always updated correctly. Otherwise, $o$ is added to $D_n$ in the loop in line \ref{line:execute_incremental_out_loop}. The loops in line \ref{line:execute_incremental_out_loop} and \ref{line:execute_incremental_in_loop} also add all operation nodes $o'$ to $D_n$ for which the result of $o'.valid(val)$ may have been impacted by the update to $val$. Thus, given the induction assumption, at the end of the loop in line \ref{line:execute_incremental_operations_loop}, we again have $\forall o \in O : \neg o.valid(val) \rightarrow o \in D \cup D_n$ and hence the induction step holds.

Since at the end of the loop in line \ref{line:execute_incremental_operations_loop}, $D = \emptyset$, we know that $\forall o \in O : \neg o.valid(val) \rightarrow o \in D_n$. Because at the end of the iteration of the loop in line \ref{line:execute_incremental_main_loop}, the set $D$ is replaced by $D_n$, the induction step for (1) holds.

Since the loop in line \ref{line:execute_incremental_main_loop} is only left when $D = \emptyset$ after the replacement with $D_n$, we know that, if the algorithm terminates, $\forall o \in O : o.valid(val)$.
\end{proof}

Similar to Algorithm \ref{algo:execute_incremental_order}, the algorithm also yields correct results if all employed $update$ procedures are at least conditionally correct, the \eGDN{}'s nodes do not share output slots that are not also input slots to all sharing nodes, and there are no deltas for an output slot of a node that is not simultaneously an input slot.

\begin{corollary}
Assuming that $\forall o_1, o_2 \in O : o_1 \neq o_2 \rightarrow \forall s \in out(o_1) \cap out(o_2) : s \in in(o_1) \wedge s \in in(o_2)$ and $\forall o \in O : \forall s \in out(o) : o.\Delta[s] = \emptyset$, for inputs $G = (O, S, E, s, t)$ and $val$, if Algorithm \ref{algo:execute_incremental} terminates, all employed $update$ procedures are conditionally correct and non-recursive, and if the valuation function before the deltas cached in $G$ was consistent with the semantics of $G$'s operation nodes, it produces a final valuation function $val$ such that $\forall o \in O : o.valid(val)$.
\end{corollary}

\begin{proof}
From the additional assumptions regarding output slots of $G$'s operation nodes, it follows directly that the condition in the definition of conditional correctness is never violated. Thus, the statement from Theorem \ref{the:correctness_incremental} also applies for the case of conditionally correct $update$ functions.
\end{proof}

\subsection{Development with \eGDN{}s}

Since Algorithm \ref{algo:find_incremental_order} considers only the \eGDN{} structure and no concrete slot contents, it can be employed as a tool for statically analyzing \eGDN{}s. In particular, via the algorithm, configurations of slots with modified contents can be analyzed regarding termination of a corresponding \eGDN{} execution. For instance, the algorithm can be used to check whether termination is guaranteed if a specific individual model is modified.

If this is the case for all user-editable models, a conservative approach that always guarantees terminating \eGDN{} executions and correct results while avoiding the exponential effort of executing the analysis for every combination of user-editable models would be enforcing a \textbf{direct propagation policy}. Under this policy, after modifying a single model, the corresponding changes would immediately be propagated to restore consistency. Only after that, the modification of a different model would be permitted.

Furthermore, Algorithm \ref{algo:find_incremental_order} can be adapted to return the set of slots $closure_\Delta(S_i)$ that may be automatically modified by \eGDN{} operations if the \eGDN{} were to be executed via Algorithm \ref{algo:execute_incremental_order} with initially modified slots $S_i$. This enables collaborative development of a network of models managed via an \eGDN{} with guaranteed termination and conflict-free consistency restoration via a \textbf{propagation closure locking policy}. For a set of already modified slots $S_\Delta$, this policy would only allow modification of the contents of another slot $s$ if, for the set $S_\Delta \cup \{s\}$, Algorithm \ref{algo:find_incremental_order} produces an execution order. Furthermore, to guarantee that no user edits are overwritten, the policy would check whether $S_\Delta \cup \{s\} \cap closure_\Delta(S_\Delta \cup \{s\}) = \emptyset$. Note that the restrictions of this policy would also apply in the case where the same user wants to edit the contents of multiple slots.

Since Algorithm \ref{algo:execute_incremental} does not guarantee termination, careful consideration is required if an \eGDN{} cannot be executed via Algorithm \ref{algo:execute_incremental_order}. However, if developers are confident that their \eGDN{} is guaranteed to terminate despite cyclical dependencies at the structural level, Algorithm \ref{algo:execute_incremental} can be used as a fallback option for \eGDN{} execution.

The presented algorithms also enable the treatment of sub-\eGDN{}s as operation nodes of a parent \eGDN{}, as they essentially provide a realization of the required $update$ procedure.

\section{Implementation}\label{sec:egdns_implementation}

We have prototypically implemented a number of concrete example operation node types for the construction of \eGDN{}s for usage in the context of the Eclipse Modeling Framework (EMF) \cite{emf}. In addition to listing the implemented operations' names, Table \ref{tab:egdn_nodes} also provides brief descriptions of their behavior. Table \ref{tab:egdn_nodes_properties} characterizes our implementations in terms of the properties defined in this report.

\textbf{Non-recursiveness: }The $update$ procedure of TGG Snychronisation operations is non-recursive if the slots containing source, target, and correspondence model are distinct. The non-recursiveness of composite nodes depends on the exact composition of the sub-\eGDN{}. All other nodes' $update$ procedures are only guaranteed to be non-recursive if their input and output slots are distinct. The checkmark symbol \checkmark indicates non-recursive $update$ procedures under this assumption.

\textbf{Potential Population/Update Directions: } The potential update directions of the TGG Synchronization ($\leftrightarrow$) can be characterized as follows (under the assumption of distinct slots for source, target, and correspondence model): If the set of considered input slots is empty, no modifications will be made to the contents of any output slot. If the set of considered input slots contains only the source model, the operation will only modify the target model and correspondence model and vice-versa. In all other cases, all models may be modified. The potential update directions of composite nodes are determined by the exact structure of the sub-\eGDN{}. All other nodes may modify the contents of all of their output slots for any set of considered input slots. The potential update direction function $dir_\Delta$ of all example nodes is union monotonic.

\textbf{Correctness: }The $update$ procedures of all operation implementations are only conditionally correct. Effectively, this means that operations may not share output slots and no user edits are allowed to output slots of operation nodes, unless the shared or edited output slot is also simultaneously an input slot of the concerned operation nodes.

\textbf{Incrementality: }The checkmark symbol \checkmark indicates a fully incremental $update$ procedure under the assumption of ideal data structures. Also operations which are listed as not fully incremental support incremental execution to some extent. The degree of incrementality depends on the operation and its concrete inputs. Naturally, our implementation of the Expression node is only fully incremental if the evaluation of the considered expression has a runtime complexity in $O(1)$. In the case of the Pattern Matching and TGG Synchronization node, a fully incremental execution can be achieved for certain input models and patterns respectively TGGs. The degree of incrementality of the execution of an \eGDN{} or sub-\eGDN{} depends on which slots are designated the \eGDN{}'s interface slots, as well as the contained operation nodes and their composition. While the Group Expression node also has a partially incremental update procedure, due to the handling of collections via the employed OCL-interpreter, a fully incremental execution is usually not possible.

As interfaces between these operations, that is, slot nodes, our implementation employs regular EMF models for model slots and hash-based indices for assignment slots. While the choice of hash-based indices over array-based indices means that the theoretically fully incremental operation implementations may not be fully incremental in conjunction with our slot implementations, hash-based data structures are usually preferable in practice due to their lower memory footprint and exhibit acceptable performance in most scenarios.

\begin{table}
\centering
\begin{tabularx}{\textwidth}{| l | X |}
\hline
\textbf{Name} & \textbf{Description}\\
\hline
\multicolumn{2}{| l |}{\textbf{RETE Nodes}} \\
\hline
Node Input & extracts individual nodes of a given type from a model \\
\hline
Edge Input & extracts individual edges of a given type from a model \\
\hline
Join & performs a natural join of assignments stored in two input assignment slots \\
\hline
Anti-Join & performs an anti-join of a left input assignment slot against a right input assignment slot \\
\hline
\multicolumn{2}{| l |}{\textbf{GDN Nodes}} \\
\hline
Pattern Matching & finds matches for a given pattern into a model; supports additional constraints formulated in OCL \cite{ocl}; supports constraints regarding the existence/absence of matches for other patterns via dependencies to related assignment slots \\
\hline
\multicolumn{2}{| l |}{\textbf{Property Computation Nodes}} \\
\hline
Expression & computes the value of an OCL \cite{ocl} expression for individual assignments \\
\hline
Group Expression & computes the value of an OCL \cite{ocl} expression for collections of assignments grouped by certain variables \\
\hline
Group Count & counts the number of assignments in collections of assignments grouped by certain variables \\
\hline
Group Sum & computes the sum of numerical values of a specific variable in collections of assignments grouped by certain variables \\
\hline
\multicolumn{2}{| l |}{\textbf{Transformation Nodes}} \\
\hline
TGG Sync. ($\rightarrow$) & performs unidirectional model synchronization of changes from a source to a target and associated correspondence model via a triple graph grammar \cite{Sch94_2_ref} \\
\hline
TGG Sync. ($\leftrightarrow$) & performs bidirectional model synchronization of changes between a source, target, and associated correspondence model via a triple graph grammar \cite{Sch94_2_ref} \\
\hline
\multicolumn{2}{| l |}{\textbf{Composite Nodes}} \\
\hline
\eGDN{} & executes a sub-\eGDN{} to update the contents of exposed slots via Algorithm \ref{algo:execute_incremental_order} or Algorithm \ref{algo:execute_incremental} \\
\hline
\end{tabularx}
\caption{Example \eGDN{} node types}\label{tab:egdn_nodes}
\end{table}

\begin{table}
\centering
\begin{tabular}{| l | c | c | c | c |}
\hline
\textbf{Name} & \textbf{Non-recursive} & \textbf{Directions} & \textbf{Correct} & \textbf{Incremental}\\
\hline
\multicolumn{5}{| l |}{\textbf{RETE Nodes}} \\
\hline
Node Input & \checkmark & all & cond. & \checkmark \\
\hline
Edge Input & \checkmark & all & cond. & \checkmark \\
\hline
Join & \checkmark & all & cond. & \checkmark \\
\hline
Anti-Join & \checkmark & all & cond. & \checkmark \\
\hline
\multicolumn{5}{| l |}{\textbf{GDN Nodes}} \\
\hline
Pattern Matching & \checkmark & all & cond. & (\checkmark) \\
\hline
\multicolumn{5}{| l |}{\textbf{Property Computation Nodes}} \\
\hline
Expression & \checkmark & all & cond. & (\checkmark) \\
\hline
Group Expression & \checkmark & all & cond. & $\sim$ \\
\hline
Group Count & \checkmark & all & cond. & \checkmark \\
\hline
Group Sum & \checkmark & all & cond. & \checkmark \\
\hline
\multicolumn{5}{| l |}{\textbf{Transformation Nodes}} \\
\hline
TGG Sync. ($\rightarrow$) & \checkmark & all & cond. & (\checkmark) \\
\hline
TGG Sync. ($\leftrightarrow$) & \checkmark & * & cond. & (\checkmark) \\
\hline
\multicolumn{5}{| l |}{\textbf{Composite Nodes}} \\
\hline
\eGDN{} & ? & ? & cond. & (\checkmark) \\
\hline
\end{tabular}
\caption{Properties of example \eGDN{} node types}\label{tab:egdn_nodes_properties}
\end{table}

Figure \ref{fig:sample_egdn_simple} shows a more complex version of the example \eGDN{} from Figure \ref{fig:sample_egdn_simple} that can be realized using the introduced example \eGDN{} nodes from Table \ref{tab:egdn_nodes}. The transformation from class diagram to abstract syntax graph is now concretely realized via a unidirectional TGG Synchronization. The query operation that was previously represented by a single query node is decomposed into a complex network of subqueries. This sub-\eGDN{} consists of two Pattern Matching nodes labeled ``$x \rightarrow y$'' that look for primitive patterns consisting of a single edge, one Group Count and one Group Sum node visualized as nodes labeled ``COUNT (X)'' respectively ``SUM (X)'', and a Join node labeled ``$\bowtie$''. Alternatively, the Pattern Matching nodes could also be realized as Edge Inputs.

\begin{figure}
\includegraphics[width=\textwidth]{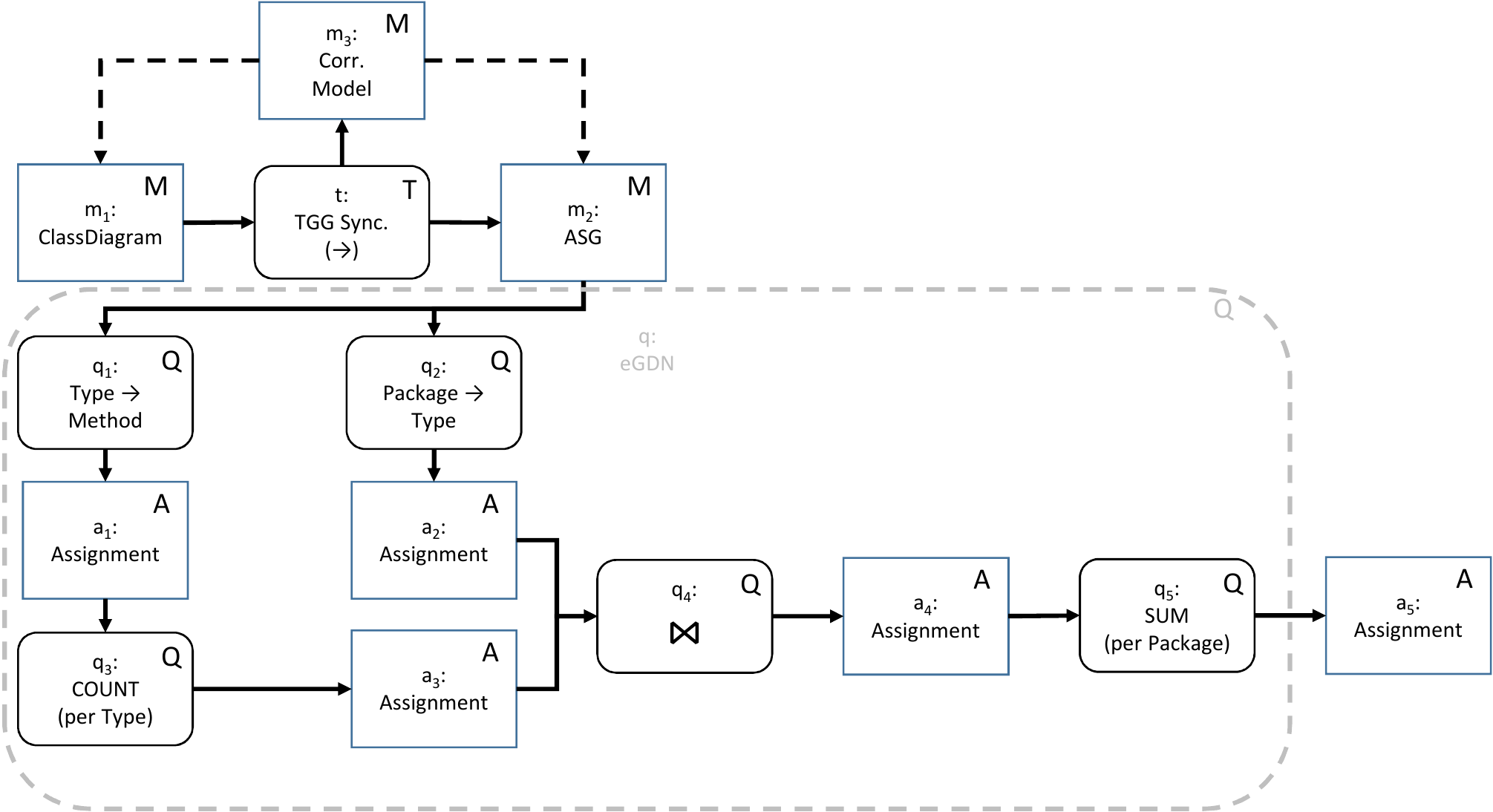}
\caption{Complex example \eGDN{}} \label{fig:sample_egdn_complex}
\end{figure}


\section{Evaluation}\label{sec:evaluation}

In this chapter, we report on an initial empirical evaluation based on our prototypical implementation. Moreover, we describe how \eGDN{}s can be employed in a typical application scenario, evaluating the developed approach with respect to the requirements from Chapter \ref{sec:requirements}.

\subsection{Evaluation of Performance}

For an initial empirical evaluation of the proposed approach, we perform an experiment inspired by an application scenario from the software development domain, where an evolving class diagram serves as the basis for generating object-oriented code, which is subsequently analyzed to compute code metrics.

Therefore, we have implemented a simple model transformation from Ecore models \cite{emf} to Java abstract syntax graphs \cite{bruneliere2010modisco} via a triple graph grammar. For each class in the class diagram, the transformation creates an interface in the Java abstract syntax graph in a first package, along with an implementation class in a second package. Also, for each attribute of a class in the class diagram, the transformation creates a corresponding field and associated getter and setter methods in the corresponding interface and class in the abstract syntax graph.

In addition, we have realized a model query composed of several subqueries, which counts the number of methods in all types of a Java package. The transformation and query are integrated into an \eGDN{}, which yields the structure displayed in Figure \ref{fig:sample_egdn_complex}.

Using our prototypical implementation, which is available under \citeOwn{implementation}, we assign a real-world Ecore model \cite{bruneliere2010modisco} to the class diagram model slot and perform an initial population of the remaining slots via Algorithm \ref{algo:execute_incremental_order}. To evaluate the scalability of the \eGDN{}, we then apply a number of synthetic updates to the model in the class diagram slot, each of which adds an attribute to each class in the model, and measure the time required for the \eGDN{} to process each such update via Algorithm \ref{algo:execute_incremental_order} (``INCREMENTAL''). We compare this to a baseline, where instead, we perform a full recomputation of both the model transformation's and the query's results via non-incremental implementations of the corresponding operations (``BATCH'').\footnote{All experiments were performed on a Linux SMP Debian 4.19.67-2 machine with Intel Xeon E5-2630 CPU (2.3\,GHz clock rate) and 386\,GB system memory running OpenJDK version 11.0.6. Reported execution time measurements correspond to the mean execution time of 10 runs of the respective experiment.}

Figure \ref{fig:execution_times} displays the execution times for the first 30 updates. After an initial phase comprising the first 5 updates, where execution time decreases from update to update, the execution time for processing an update to the class diagram via the strategy INCREMENTAL does not change much. In particular, there does not seem to be any trend of increasing execution time related to the growth of the class diagram as additional updates are being performed. In contrast, the execution time of BATCH increases from update to update as the class diagram grows. While it starts out similar to the execution time of INCREMENTAL (larger by factor 1.6), by update 30 the execution time of BATCH has increased to factor 80 compared to the execution time of INCREMENTAL.

\begin{figure}
\centering
\includegraphics[width=0.75\textwidth]{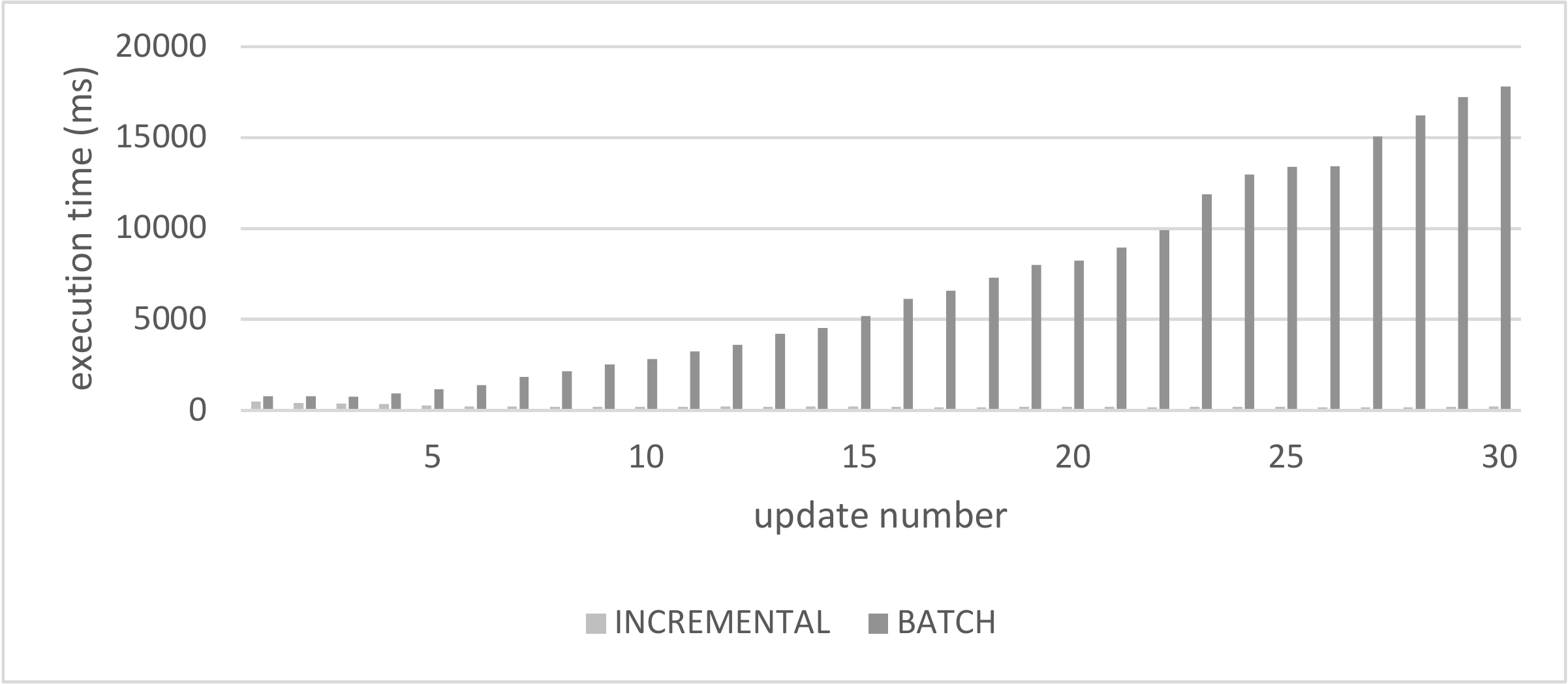}
\caption{Execution time measurements for class diagram updates}\label{fig:execution_times}
\end{figure}

The measurements thus indicate that incremental \eGDN{} execution via INCREMENTAL efficiently handles updates to the class diagram, in the sense that execution time only seems to depend on the actual changes rather than the size of the model, indeed affording incrementality. Therefore, \eGDN{}s seem to constitute a suitable formalism for a scalable, modular and incremental realization of networks of model operations for this scenario. The decreasing execution times per update during the initial phase of the experiment can likely be attributed to warming-up effects of the Java virtual machine.

The internal validity of our results is mostly threatened by unexpected behavior of the Java virtual machine, most notably garbage collection. To mitigate such effects, the reported execution time measurements were obtained as the arithmetic mean of multiple runs of the experiment, with the standard deviation of the overall execution time always below 5\% of the overall execution time.

The synthetic updates used in the experiment pose a threat to external validity. However, the experiment is inspired by a real-world application scenario and uses a real-world model as its basis and demonstrates the applicability of the \eGDN{} approach in this scenario. The synthetic updates only serve the purpose of allowing a systematic evaluation of our technique's scalability. We hence do not make any quantitative claims regarding our approach in practical application scenarios, but merely consider our experimental results as an indicator for the presented approach's potential. We furthermore do not make claims regarding the generalizability of the approach to other application domains, which would require further evaluation and is left for future work.

\subsection{Evaluation of Applicability} \label{sec:evaluation_applicability}

In order to investigate the applicability of the developed technique, we consider the following extended example scenario that requires global model management: A class diagram, adhering to a metamodel similar to the one displayed in Figure \ref{fig:sample_graph}, is used to model the structure of a software system under development by means of classes contained in packages. Classes may contain methods, which may in turn reference classes as the method's return type. OCL expressions in a separate model are used to describe the behavior of some of the class diagram's methods. Therefore, the OCL model has its own representation of types corresponding to the classes in the class diagram. This correspondence is captured by means of a linking model, which simply contains dedicated link vertices. A link vertex can either have edges to a class from the class diagram and the corresponding type from the OCL model or edges to a method in the class diagram and the corresponding expression, that is, implementation, in the OCL model.

We consider the following use cases for this setup:

\begin{itemize}
\item \textbf{Consistency Checking:} The developers want to run automatic and incremental consistency checks that verify that the return type of a method in the class diagram matches the corresponding type of the method's OCL-implementation. The developed consistency check should also work for similar setups that use a different expression language than OCL for the method implementations. An implementation for this use case thus requires a solution satisfying the requirements \requirement{1.1}, \requirement{1.2}, \requirement{2.2}, and \requirement{2.3}.
\item \textbf{Code Generation:} The developers want to automatically and incrementally generate Java code in the form of an ASG from the class diagram. In addition, Java implementations for the class diagram's methods should be generated from the methods' OCL implementations. In the end, the resulting Java code fragments for the two models should be integrated and analyzed for some code metrics. An implementation for this use case thus requires a solution satisfying the requirements \requirement{2.1}, \requirement{2.2}, and \requirement{2.3}.
\item \textbf{Megamodel Reuse:} After developing the automatic consistency checking and code generation, the developers want to reuse the same two operations in another project with a similar set of models. An implementation for this use case thus requires a solution satisfying the requirements \requirement{3.1.2}, \requirement{3.2}, \requirement{3.3}, and \requirement{3.4}.
\end{itemize}

In order to allow global model management for all three use cases, a solution also has to enable the modeling of a network of different kinds of model operations over a set of potentially integrated models, that is, a solution has to satisfy requirement \requirement{3.1.1}.

Figures \ref{fig:sample_egdn_consistency}, \ref{fig:sample_egdn_codegen_and_analysis}, and \ref{fig:sample_egdn_megamodel_modules} visualize example \eGDN{} implementations for the \textbf{Consistency Checking}, \textbf{Code Generation}, and \textbf{Megamodel Reuse} use cases, respectively.

As displayed in Figure \ref{fig:sample_egdn_consistency}, the \textbf{Consistency Checking} use case is realized via four Pattern Matching query nodes that extract certain simple patterns from the base models and make them accessible in a generalized format via assignment slots. Then, a complex query operation, which is composed of three Join query operations and an Anti-Join query operation (labelled $\rhd$), realizes the actual consistency check by finding all the combinations of a method from the class diagram, its implementation from the OCL model, and the associated return class respectively expression type, where the return class and expression type do not correspond. Thus, the \eGDN{}-based approach in this case fulfills the requirements \requirement{1.1} and \requirement{2.2}, as it implements a consistency check over a set of models integrated via integration links. Furthermore, the resulting implementation is reusable for different modeling languages that offer similar functionality via the generic interface provided by the assignment slots $a_1$, $a_2$, $a_3$, and $a_7$, satisfying requirement \requirement{1.2}. The example \eGDN{} also demonstrates how more complex model operations can be composed from simpler operations, satisfying requirement \requirement{2.1}, and provides an incremental execution scheme for these operations, satisfying requirement \requirement{2.3}. In particular, via Algorithm \ref{algo:find_incremental_order}, it can be verified that Algorithm \ref{algo:execute_incremental_order} provides a means of executing the \eGDN{} that guarantees both correct results and termination for changes to any combination of the three base models.

The \eGDN{} shown in Figure \ref{fig:sample_egdn_codegen_and_analysis} realizes the \textbf{Code Generation} use case via a combination of two TGG Synchronizations that translate the class diagram and OCL model into Java ASGs. The two Java models are integrated via a dedicated linking model, which is produced by a unidirectional model transformation from the original linking model and the correspondence models created by the TGG Synchronizations. Finally, query operations can be executed over the ASGs to compute code metrics. Using Algorithm \ref{algo:find_incremental_order} to analyze the \eGDN{}, it can be determined that terminating execution via Algorithm \ref{algo:execute_incremental_order} can be guaranteed for changes to any combination of the three base models. This example shows how the \eGDN{} provides a unified, modular notion of model operations along with an incremental execution scheme and demonstrates the composition of model operations, satisfying requirements \requirement{2.1}, \requirement{2.2}, and \requirement{2.3}.

Together, the \eGDN{}s in Figure \ref{fig:sample_egdn_consistency} and \ref{fig:sample_egdn_codegen_and_analysis} also illustrate how \eGDN{}s can be used as a megamodeling language, supporting different kinds of model operations, including model properties (like the metrics computed in the Code Generation use case), model consistency (like the consistency condition in the Consistency Checking use case), and model transformation and synchronization (like the transformation and synchronizations in the Code Generation use case). It also shows how integration views (like the cross-model consistency query results in slot $a_8$ in Figure \ref{fig:sample_egdn_consistency}) and traceability links (like the correspondence models produced by the TGG Synchronizations) can be represented in the language. While not present in the example \eGDN{}s, the class diagram and OCL metamodel are models themselves and could simply be made explicit by including them in dedicated model slots. Well-formedness conditions for metamodels or regular models can be realized and treated as regular query operations. \eGDN{}s thus satisfy the requirement \requirement{3.1.1}.

Finally, the \eGDN{} realization of the \textbf{Megamodel Reuse} use case in Figure \ref{fig:sample_egdn_megamodel_modules} considers the \eGDN{}s from Figure \ref{fig:sample_egdn_consistency} and \ref{fig:sample_egdn_codegen_and_analysis} as operation nodes in an overarching \eGDN{}. This exemplifies how \eGDN{}s offer modularity and incrementality at the megamodel level by considering sub-\eGDN{}s as regular operations that can be executed via the general execution scheme, which also permits the accumulation of several changes before execution. The example thus illustrates the satisfaction of requirements \requirement{3.1.2}, \requirement{3.2}, and \requirement{3.4}. The \eGDN{} also demonstrates how slots act as interfaces for these megamodel operations, satisfying requirement \requirement{3.3}.

Thus, \eGDN{}s can be employed to realize the functionality required by the three example use cases, satisfying the requirements regarding model operations, modeling languages integration, and megamodels introduced in Section \ref{sec:requirements} in this scenario.

Table \ref{tab:requirements_coverage} summarizes the coverage of the requirements by the example use cases and the \eGDN{} approach, with ``$\circ$'' denoting that the realization of a use case relates to a requirement and ``\checkmark'' indicating that a requirement is satisfied by \eGDN{}s in this scenario.

\begin{figure}
\centering
\includegraphics[width=\textwidth]{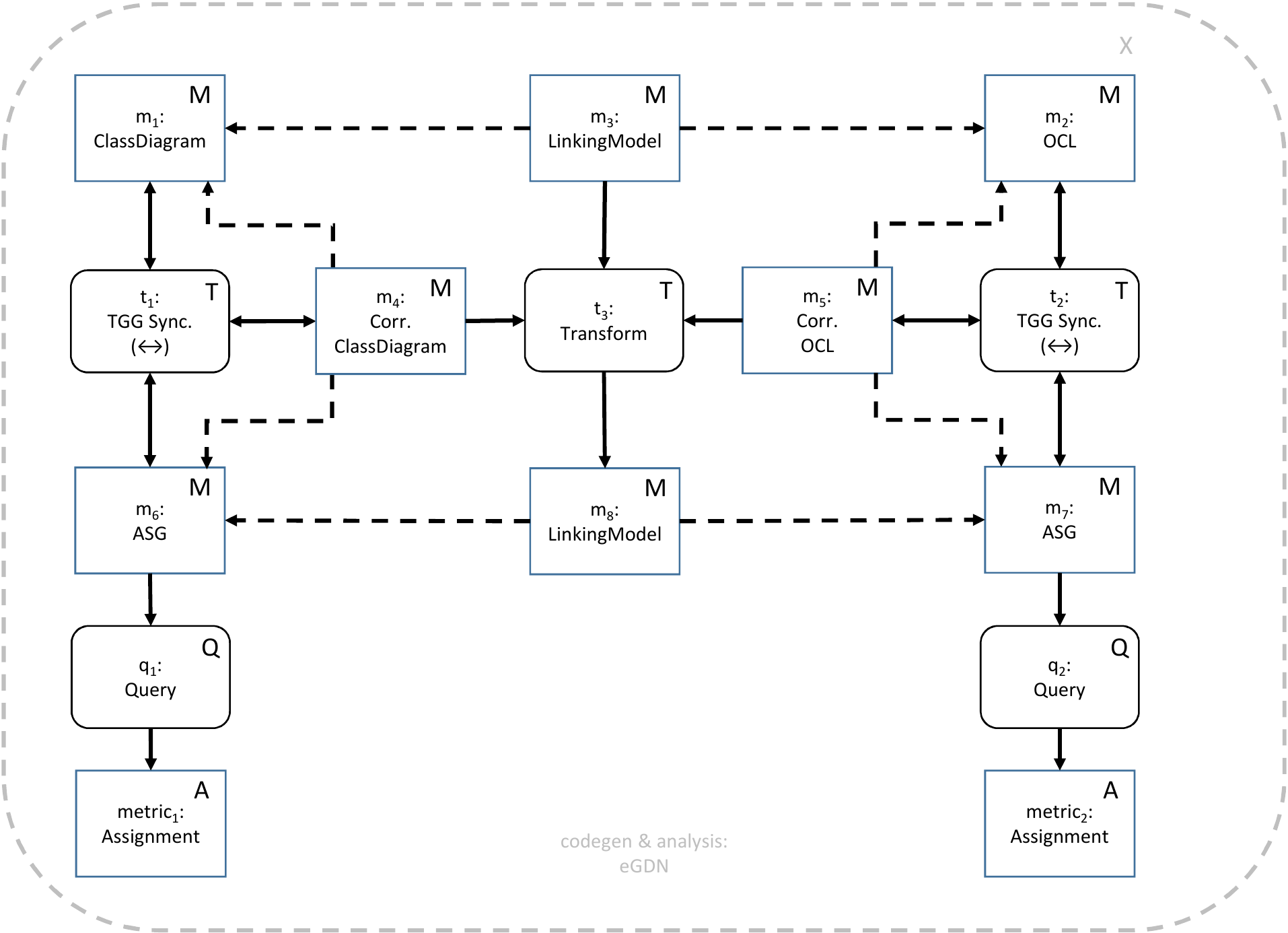}
\caption{Sample \eGDN{} realizing the Code Generation use case}\label{fig:sample_egdn_codegen_and_analysis}
\end{figure}

\begin{figure}
\centering
\includegraphics[width=\textwidth]{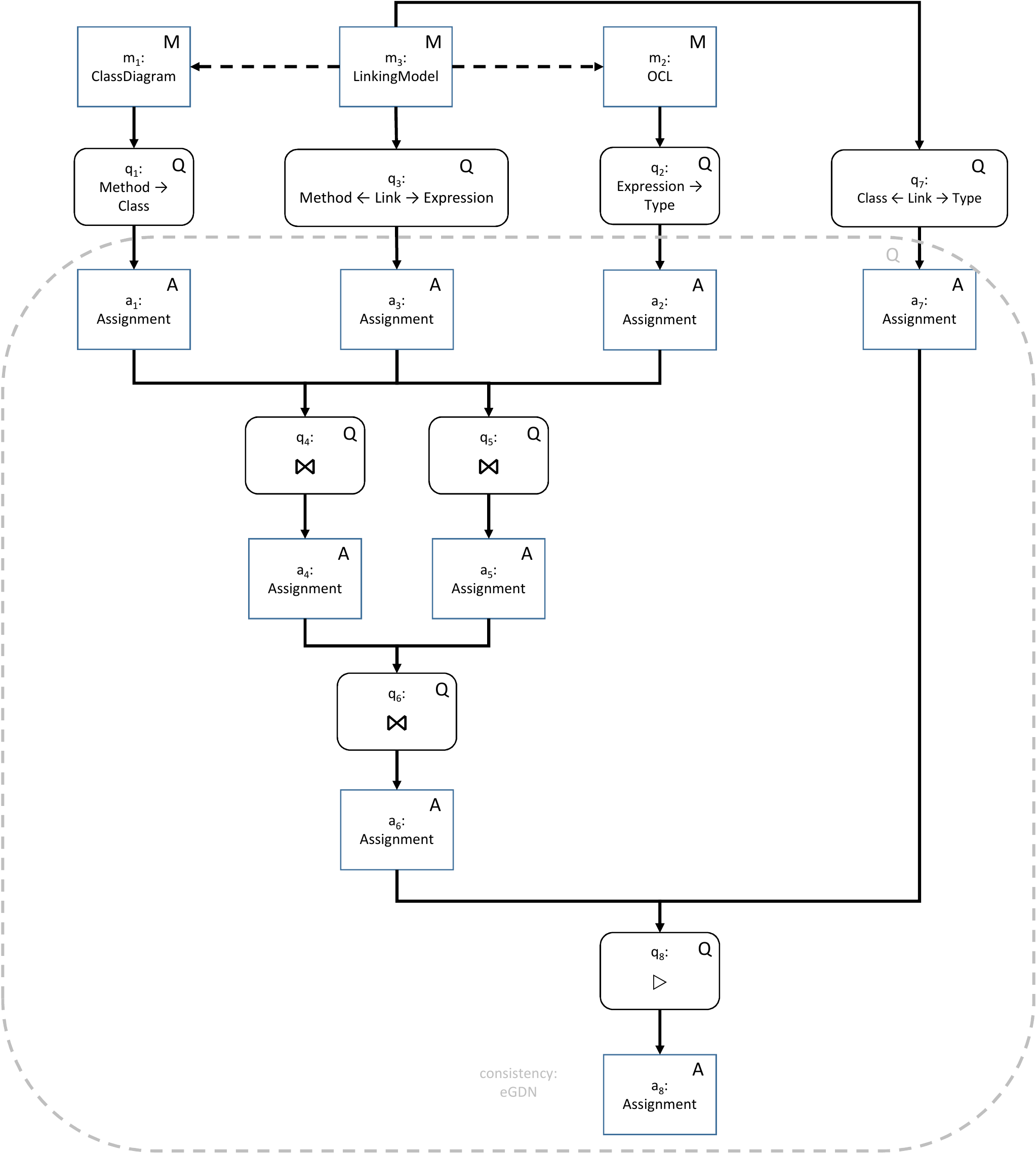}
\caption{Sample \eGDN{} realizing the Consistency Checking use case}\label{fig:sample_egdn_consistency}
\end{figure}

\begin{figure}
\centering
\includegraphics[width=\textwidth]{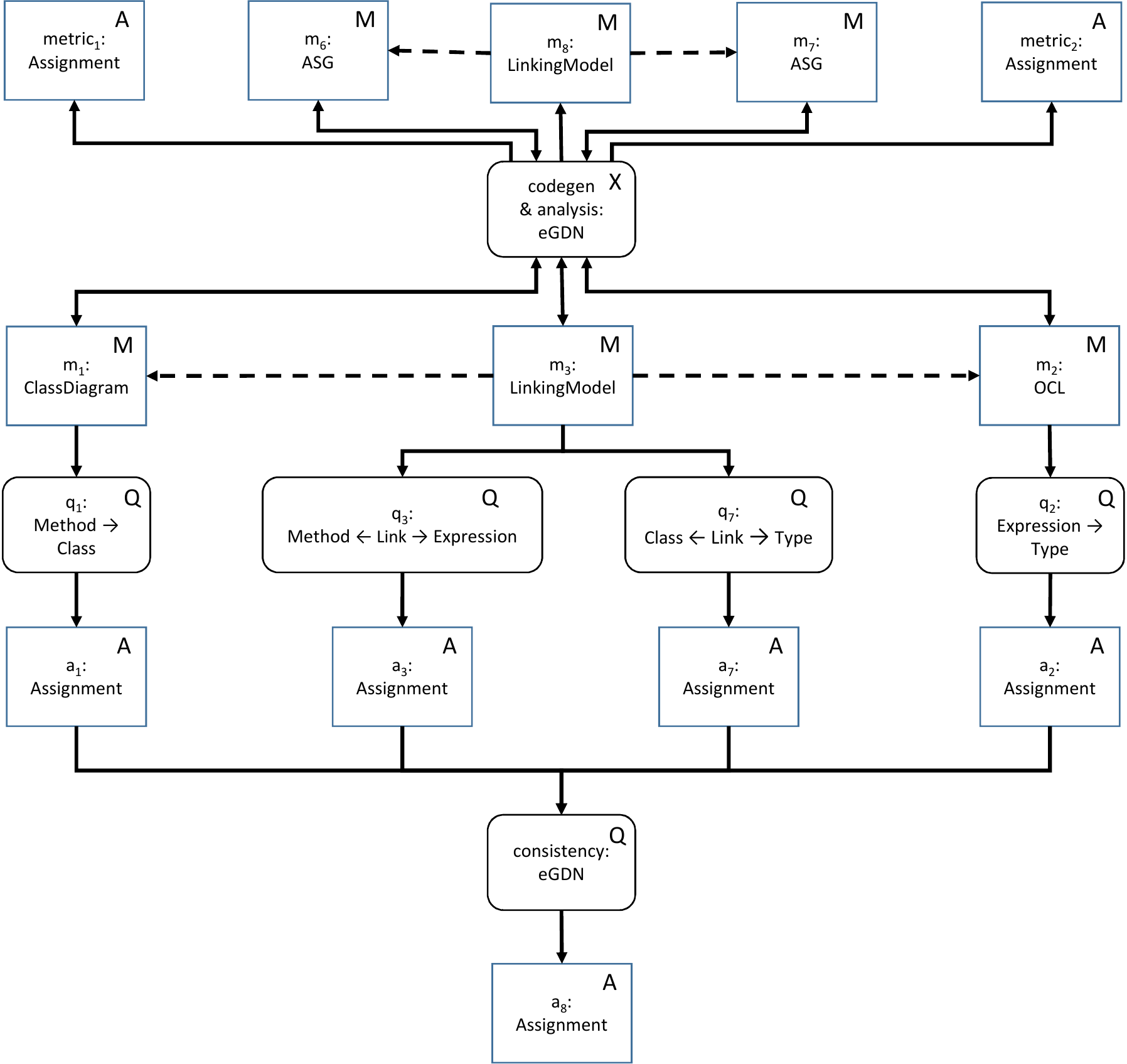}
\caption{Sample \eGDN{} realizing the Megamodel Reuse use case}\label{fig:sample_egdn_megamodel_modules}
\end{figure}

\begin{table}
\centering
\begin{tabular}{ l | c  c  c  c }
 & \rot{Consistency Checking} & \rot{Code Generation} & \rot{Megamodel Reuse} & \rot{\eGDN{}s} \\
\hline
\requirement{1.1}: modeling languages integration & $\circ$ & & & \checkmark \\
\requirement{1.2}: interfaces for embedding of modeling languages & $\circ$ & & & \checkmark \\
\requirement{2.1}: composition of model operations & & $\circ$ & & \checkmark \\
\requirement{2.2}: model operations over integrated models & $\circ$ & $\circ$ & & \checkmark \\
\requirement{2.3}: execution scheme for model operations & $\circ$ & $\circ$ & & \checkmark \\
\requirement{3.1.1}: megamodeling language & $\circ$ & $\circ$ & $\circ$ & \checkmark \\
\requirement{3.1.2}: megamodel operation module concept & & & $\circ$ & \checkmark \\
\requirement{3.2}: robust megamodel execution scheme & & & $\circ$ & \checkmark \\
\requirement{3.3}: megamodel interfaces & & & $\circ$ & \checkmark \\
\requirement{3.4}: asynchronous megamodel execution scheme & & & $\circ$ & \checkmark \\
\end{tabular}
\caption{Coverage of requirements from Chapter \ref{sec:requirements}} \label{tab:requirements_coverage}
\end{table}

\newpage


\section{Conclusion}\label{sec:conclusion}

In this report, we have developed a further generalization of the GDN mechanism called \eGDN{}s, which enables the modular and incremental construction and execution of complex networks of model operations, including model properties, model consistency, model transformation and model synchronization. In addition to a formal definition of \eGDN{}s, we have provided incremental algorithms for their execution. Moreover, we have presented a number of example \eGDN{} nodes that we have prototypically implemented in order to perform an initial empirical evaluation of the approach regarding scalability. Our experiments, which are based on an application scenario from the software development domain, indicate that the introduced technique can be employed to realize efficient Global Model Management. Moreover, we have conceptually evaluated our approach against identified requirements of global model management solutions.

In future work, we plan to perform a more extensive evaluation with respect to both expressiveness and performance of \eGDN{}s in real application scenarios. This may also involve the implementation of additional types of \eGDN{} nodes and may ultimately result in the implementation of true tool support for the specification and execution of \eGDN{}s. Furthermore, we will investigate how the presented concepts can be extended to the case of evolving modeling landscapes that consist of multiple distinct versions. We will also explore how such an extension may help alleviate problems such as potential infinite loops or overwriting of user edits in \eGDN{} execution via the derivation of additional versions.


\subsection*{Acknowledgements}

This work was developed mainly in the course of the project modular and incremental Global Model Management (project number 336677879), which is funded by the Deutsche Forschungsgemeinschaft.

\printbibliography
\end{document}